\documentclass[sigconf]{acmart}

\AtBeginDocument{%
  \providecommand\BibTeX{{%
    \normalfont B\kern-0.5em{\scshape i\kern-0.25em b}\kern-0.8em\TeX}}}

\copyrightyear{2023}
\acmYear{2023}
\setcopyright{acmcopyright}\acmConference[WSDM '23]{Proceedings of the Sixteenth ACM International Conference on Web Search and Data Mining}{February 27-March 3, 2023}{Singapore, Singapore}
\acmBooktitle{Proceedings of the Sixteenth ACM International Conference on Web Search and Data Mining (WSDM '23), February 27-March 3, 2023, Singapore, Singapore}
\acmPrice{15.00}
\acmDOI{10.1145/3539597.3570416}
\acmISBN{978-1-4503-9407-9/23/02}

\usepackage[utf8]{inputenc} 
\usepackage[T1]{fontenc}    
\usepackage{hyperref}       
\usepackage{url}            
\usepackage{booktabs}       
\usepackage{amsfonts}       
\usepackage{nicefrac}       
\usepackage{microtype}      
\usepackage{xcolor}         
\usepackage{natbib}
\usepackage{stfloats}

\usepackage{latexsym}
\usepackage{soul}
\usepackage{amsmath}
\usepackage{booktabs}
\usepackage{amsthm}
\usepackage{thmtools}
\usepackage{thm-restate}

\usepackage[utf8]{inputenc} 
\usepackage[T1]{fontenc}    
\usepackage{url}            
\usepackage{booktabs}       
\usepackage{amsfonts}       
\usepackage{nicefrac}       
\usepackage{microtype}      
\usepackage[scientific-notation=true]{siunitx}

\newtheorem{definition}{Definition}

\DeclareMathOperator*{\argmax}{argmax}

\newcommand{\E}{\mathbb{E}}
\newcommand{\I}{\mathbb{I}}
\newcommand{\R}{\mathbb{R}}

\newcommand{\cA}{\mathcal{A}}

\newcommand{\cC}{\mathcal{C}}

\newcommand{\cF}{\mathcal{F}}

\newcommand{\cL}{\mathcal{L}}
\newcommand{\cP}{\mathcal{P}}
\newcommand{\cR}{\mathcal{R}}

\newcommand{\cU}{\mathcal{U}}

\newcommand{\ap}{{\rm ap}}

\newcommand{\OPT}{{\rm OPT}}
\newcommand{\EFF}{{\rm EFF}}
\newcommand{\ERP}{{\rm ERP}}
\newcommand{\ERPN}{{\rm ERPV}}
\newcommand{\EFD}{{\rm EFD}}

\newcommand{\AAFF}{{\sf AA-FF}}

\newcommand{\AAIMMNaive}{{\sf AA-IMM-Naive}}
\newcommand{\AAIMMFB}{{\sf AA-IMM-FB}}
\newcommand{\AAIMMDAG}{{\sf AA-IMM-DAG}}

\newcommand{\NVRRP}{{\sf Naive-VRR-Path}}
\newcommand{\FBVRRP}{{\sf FB-VRR-Path}}
\newcommand{\DAGVRRP}{{\sf DAG-VRR-Path}}

\newcommand{\NodeEdgeSelection}{{\sf NodeEdgeSelection}}

\newcommand{\AAIMM}{\mbox{\sf AAIMM}}

\usepackage{helvet}  
\usepackage{courier}  
\usepackage{subfig}
\usepackage{enumitem}
\usepackage{multirow}
\usepackage[linesnumbered,ruled,vlined]{algorithm2e}
\usepackage{mathtools}
\usepackage{balance}
\usepackage{flushend}

\usepackage{ifthen}

\newcommand{\alg}[1]{\textnormal{\textsf{#1}}}

\newcommand{\newalg}[1]{{#1}}
\SetNlSty{bfseries}{\color{black}}{}

\SetCommentSty{mycommfont}

\newcommand{\compilehidecomments}{true}
\ifthenelse{\equal{\compilehidecomments}{false} }{%
	\newcommand{\ruibin}[1]{}
	\newcommand{\wei}[1]{}
}{
	\newcommand{\ruibin}[1]{{\color{red} [\text{Ruibin:} #1]}}
	\newcommand{\wei}[1]{{\color{blue} [\text{Wei:} #1]}}
}

\newcommand{\compilefullversion}{true}
\ifthenelse{\equal{\compilefullversion}{false}}{%
	\newcommand{\OnlyInFull}[1]{}
	\newcommand{\OnlyInShort}[1]{#1}
}{%
	\newcommand{\OnlyInFull}[1]{#1}%
	\newcommand{\OnlyInShort}[1]{}%
}%

\begin{document}

\title{Scalable Adversarial Attack Algorithms on Influence Maximization}




\author{Lichao Sun}
\affiliation{%
 \institution{Lehigh University}
 \city{Bethlehem}
 \state{PA}
 \country{USA}}
\email{lis221@lehigh.edu}

\author{Xiaobin Rui}
\authornote{Corresponding author}
\affiliation{%
  \institution{China University of Mining and Technology}
  \city{Xuzhou}
  \state{Jiangsu}
  \country{China}
}
\email{ruixiaobin@cumt.edu.cn}

\author{Wei Chen}
\affiliation{%
  \institution{Microsoft Research Asia}
  \city{Beijing}
  \country{China}}
\email{weic@microsoft.com}






\begin{abstract}
  In this paper, we study the adversarial attacks on influence maximization 
  under dynamic influence propagation models in social networks.
  In particular, given a known seed set $S$, the problem is to minimize the influence spread from $S$ by deleting a limited number of nodes and edges.
  This problem reflects many application scenarios, such as blocking virus (e.g. COVID-19) propagation in social networks by quarantine and vaccination, blocking rumor spread by freezing fake accounts, or attacking competitor's influence by incentivizing some users to ignore the information from the competitor. 
  In this paper, under the linear threshold model, we adapt the reverse influence sampling approach and provide efficient algorithms of sampling valid reverse reachable paths to solve the problem.
  We present three different design choices on reverse sampling, which all guarantee $1/2 - \varepsilon$ approximation (for any small $\varepsilon >0$) and an efficient running time.
\end{abstract}




\maketitle

\section{Introduction}
Influence maximization (IM) is the optimization problem of finding a small set of most influential nodes in a social network that generates the largest influence spread,
which has many applications such as promoting products or brands through viral marketing in social networks~\cite{domingos01,richardson02,kempe03}.
However, in real life, there are so many competitions in various scenarios with different purposes, such as attacking competitor's influence ~\cite{Leskovec07}, controlling rumor~\cite{BAA11,HeSCJ12}, or blocking the virus spread.
The adversary will block the virus propagation by quarantine and vaccination, block rumor spread by freezing fake accounts, or attack competitors' influence by incentivizing some users to ignore the information from the competitor.
All these scenarios can be modeled as the adversary trying to remove certain nodes and edges 
to minimize the influence or impact from the competitor's seed set in social networks, which we denoted as 
the adversarial attacks on influence maximization (AdvIM).

We model the AdvIM task more formally as follows.
The social influence network is modeled as a weighted network $G=(V,E,w)$, where $V$ is a set of nodes representing individuals, $E$ is a set of directed edges representing influence relationships, and $w$ is influence weights on edges.
In the beginning, we have a fixed seed set $S$, and the propagation from $S$ follows the classical linear threshold model~\cite{kempe03}.
For an attack set $A$ consisting of a mix of nodes (disjoint from $S$) and edges, we measure the effectiveness of $A$ by the {\em influence reduction} $\rho_S(A)$ it achieves, which is defined as the difference in influence spread with and without removing nodes and edges in the attack set $A$.
Then, given a node budget $q_N$ and an edge budget $q_E$, the AdvIM task is to find an attack set $A$ with at most $q_N$ nodes (excluding any seed node) and at most $q_E$ edges, 
such that after removing the nodes and edges in $A$, the influence spread of $S$ is minimized, or the {\em influence reduction} $\rho_S(A)$ is maximized.
%
%

We show that under the LT model, the influence reduction function $\rho_S(A)$ is monotone and submodular, which enables a greedy approximation algorithm.
However, the direct greedy algorithm is not efficient since it requires a large number of simulations of propagation from the seed set $S$.
In this paper, we adapt the reverse influence sampling (RIS) approach to design efficient algorithms for the AdvIM task.
Due to the nature of the problem, only successful propagation from the seed set can be potentially reduced by the attack set.
This creates a new challenge for reverse sampling.
In this paper, we present three different design choices for such reverse sampling and their theoretical analysis.
They all provide a $1/2-\varepsilon$ approximation guarantee and have different trade-offs in efficiency.
We then conduct experimental evaluations on several real-world networks and demonstrate that our algorithms achieve good influence reduction results while running much faster than existing greedy-based algorithms.
To summarize, \textit{our contributions} are:
(a) proposing the study of adversarial attacks on influence maximization problems;
(b) designing efficient algorithms by adapting the RIS approach and providing their theoretical guarantees; and
(c) conducting experiments on real-world networks to demonstrate the effectiveness and efficiency of 
our proposed algorithms.
\paragraph{Related Work.}
Domingos and Richardson first studied Influence Maximization (IM)~\cite{domingos01,richardson02}, and then IM is mathematically formulated as a discrete optimization problem by Kempe et al. \citep{kempe03},
who also formulate the independent cascade model, the linear threshold model, the triggering model, and provide a greedy approximation algorithm based on submodularity.
After that, most works focus on improving the efficiency and scalability of influence maximization algorithms~\cite{ChenWY09,ChenYZ10,WCW12,JungHC12,BorgsBrautbarChayesLucier,tang14,tang15}.
The most recent and the state of the art is the reverse influence sampling (RIS) approach~\cite{BorgsBrautbarChayesLucier,tang14,tang15,Nguyen_DSSA_2016,Tang_OPIM_2018}, and the \alg{IMM} algorithm of~\cite{tang15} is one of the representative algorithms for the RIS approach.
Some other studies look into different problems, such as competitive and complementary influence maximization~\cite{kempe07,BAA11,chen2011influence,HeSCJ12,lu2015competition,he2016joint},
adoption maximization~\cite{bhagat_2012_maximizing}, 
robust influence maximization~\cite{ChenLTZZ16,HeKempe16}, etc.
The most similar topic to this work is competitive influence maximization that aims to maximize the influence while more than one party is in the network. 

One similar work also wants to stop the influence spread in the social network \cite{khalil2014scalable}. 
However, this work has a different objective function that estimates the total influence by summing up the influence of the individual node in seed sets, which is different from the traditional influence maximization. 
They proposed the greedy approach through forward-tree simulation without time-complexity analysis. In this work, our objective function directly matches the  original influence maximization objective function.
Moreover, we propose RIS-based algorithms to overcome the efficiency issue in the forward simulation approach while also providing theoretical guarantee on both the approximation ratio and the running time.

\OnlyInShort{The full version of the paper with complete proofs and other technical details is available at ~\cite{arxivfull}.}
\section{Model and Problem Definition} \label{sec:model}


\paragraph{Adversarial Attacks on Diffusion Model.}
In this paper, we focus on the well-studied {\em linear threshold (LT) model} \cite{kempe03} as the basic diffusion model. 
A social network under the LT model is modeled as a directed influence graph $G=(V,E, w)$, where $V$ is a finite set of vertices or nodes, $E \subseteq V \times V$ is the set of directed edges connecting pairs of nodes, and	$w: E \rightarrow [0,1]$ gives the influence weights on all edges.
The diffusion of information or influence proceeds in discrete time steps $t = 0, 1, 2, \dots$. 
At time $t=0$, the {\em seed set} $S_0$ is selected to be active, and also each node $v$ independently selects a threshold $\theta_v$ uniformly at random in the range $[0, 1]$, corresponding to users' true thresholds. 
At each time $t\ge 1$, an inactive node $v$ becomes active if $\sum_{u:u\in S_{t-1},(u,v)\in E}w(u,v) \ge \theta_v$ where $S_{t-1}$ is the set of nodes activated by time $t-1$. 
The diffusion process ends when there are no more nodes activated in a time step.

Given the weights of all nodes $v\in V$, we can construct the live-edge graph $L = (V, E(L))$, where at most one of each $v$'s incoming edges is selected with probability $w(u, v)$, and no edge is selected with probability $1 - \sum_{u:(u,v)\in E} w(u, v)$.
Each edge $(u,v)\in L$ is called a {\em live edge}.
Kempe et al. \cite{kempe03} show that the propagation in the linear threshold model is equivalent to the deterministic propagation via bread-first traversal in a random live-edge graph $L$.
An important metric in a diffusion model is the {\em influence spread}, defined as the expected number of active nodes when the propagation from the given seed set $S_0$ ends, and is denoted as $\sigma(S_0)$.
Let $\Gamma(G, S)$ denote the set of nodes in graph $G$ that can be reached from the node set $S$.
Then, by the above equivalent live-edge graph model, we have $\sigma(S_0) = \E_L[|\Gamma(L, S_0)|] = \sum_L Pr[L|G] \cdot |\Gamma(L, S_0)|$, where the expectation is taken over the distribution of live-edge graphs, and  $Pr[L|G]$ is the probability of sampling a particular live-edge graph $L$ in graph $G$.
As defined above, we have:
\[
p(v, L, G)= 
\begin{cases}
w(u,v),& \text{if } \exists	u:(u,v) \in L\\
1 - \sum_{u:(u,v)\in E} w(u, v),              & \text{otherwise}
\end{cases}
\]
which is the probability of the configuration of incoming edges for node $v$ in $L$; then, the probability of a particular live-edge graph $L$ is $Pr[L|G] = \prod_{v\in V}p(v, L, G)$.
When we need to specify the graph, we use $\sigma(S_0, G)$ to represent the influence spread under graph $G$.

A set function $f:V\rightarrow \R$ is called {\em submodular} if for all $S\subseteq T \subseteq V$ and $u \in V\setminus T$, 
$f(S\cup \{u\}) - f(S) \ge f(T \cup \{u\}) - f(T)$.
Intuitively, submodularity characterizes the diminishing return property often occurring in economics and operation research.
Moreover, a set function $f$ is called {\em monotone} if for all $S\subseteq T \subseteq V$, $f(S) \le f(T)$.
It is shown in~\cite{kempe03} that influence spread $\sigma$ for the linear threshold model is a monotone submodular function.
A non-negative monotone submodular function allows a greedy solution to its maximization problem subject to a cardinality constraint, with an approximation ratio $1-1/e$, where $e$ is the base of the natural logarithm~\cite{NWF78}.
This is the technical foundation for most influence maximization tasks.



\paragraph{Adversarial Attacks on Influence Maximization.}
The classical influence maximization problem is to choose a seed set $S$ of size at most $k$ seeds to maximize the influence spread $\sigma(S, G)$.
For the Adversarial Attacks Influence Maximization (AdvIM) problem, the goal is to select at most $q_N$ nodes and $q_E$ edges to be removed, such that the influence spread on a given seed set $S$ is minimized.
Let $A$ denote a joint {\em attack set}, which contains a subset of nodes $A_N \subseteq V \setminus S$ and a subset of edges $A_E\subseteq E$, i.e. $A = A_N \cup A_E$.
Denote the new graph after removing the nodes and edges in $A$ as $G' = G \setminus A$.
We first define the key concept of influence reduction under the attack set $A$.
\begin{definition}[Influence Reduction]
Given a seed set $S$, the influence reduction under the attack set $A$, denoted as $\rho_S(A)$, is the reduction in influence spread from the original graph $G$	to the new graph $G' = G\setminus A$. 
That is, $\rho_S(A) = \sigma(S,G) - \sigma(S, G')$.
\end{definition}
Note that $S \cap A = \emptyset$, which means we cannot attack any seed node.
We can now define the main optimization task in this paper.
\begin{definition}
    The Adversarial Attacks on Influence Maximization (AdvIM) under the linear threshold model is the optimization task where the input includes the directed influence graph $G=(V,E,w)$, seed set $S$, attack node budget $q_N$ and attack edge budget $q_E$.
    The goal is to find an attack set $A$ to remove, which contains at most $q_N$ nodes (excluding the seed set $S$) and $q_E$ edges, such that the total influence reduction is maximized: $A^* = \argmax_{A: |A_N|\leq q_N, |A_E|\leq q_E} \rho_S(A)$.
\end{definition}




Before the algorithm design, we first establish the important fact that $\rho_S(A)$ as a set function is monotone and submodular.
\begin{restatable}{lemma}{lemmaone} \label{lem:lemma1} 
Influence reduction $\rho_S(A)$ for the LT model satisfies monotonicity 
and submodularity. %
\end{restatable}


\begin{proof}
The proof is based on the live-edge graph representation of the LT model.
The proof for the monotonicity property is straightforward, so we focus on the submodularity property.
Let $A_1$ and $A_2$ be two attack sets with $A_1 \subseteq A_2$, and $x \in ((V\setminus S) \cup E) \setminus A_2$.
According to the definition of the influence reduction, we have
\begin{align*}
    &\rho_S(A_1 \cup \{x\}) -\rho_S(A_1) \geq \rho_S(A_2 \cup \{x\}) -\rho_S(A_2) \\
    \Leftrightarrow &  \E_{L}[|\Gamma(L, S) \setminus \Gamma(L\setminus A_1 \setminus \{x\}, S)|] - \E_{L}[|\Gamma(L, S) \setminus \Gamma(L\setminus A_1, S)|] \\
    \geq &  \E_{L}[|\Gamma(L, S) \setminus  \Gamma(L\setminus A_2 \setminus \{x\}, S)|] - \E_{L}[|\Gamma(L, S) \setminus \Gamma(L\setminus A_2, S)|].
\end{align*}
Then it is sufficient to show that for any fixed live-edge graph $L$, for every node $v\in V\setminus S$, if $v \in \Gamma(L, S) \setminus  \Gamma(L\setminus A_2 \setminus \{x\}, S)$ but $v\notin \Gamma(L, S) \setminus \Gamma(L\setminus A_2, S)$, 
then  $v \in \Gamma(L, S) \setminus  \Gamma(L\setminus A_1 \setminus \{x\}, S)$ but $v\notin \Gamma(L, S) \setminus \Gamma(L\setminus A_1, S)$.
Now suppose that $v \in \Gamma(L, S) \setminus  \Gamma(L\setminus A_2 \setminus \{x\}, S)$ but $v\notin \Gamma(L, S) \setminus \Gamma(L\setminus A_2, S)$.
This means that $v$ cannot be reached from $S$ in $L$  after $A_2 \cup \{x\}$ are removed from $L$, 
but $v$ can be reached from $S$ if we only remove $A_2$ from $L$.
Since $A_1 \subseteq A_2$, we have that when we remove $A_1$ from $L$, $v$ can still be reached from $S$, i.e. 
$v\notin \Gamma(L, S) \setminus \Gamma(L\setminus A_1, S)$.
Since after further removing $x$, $v$ cannot be reached from $S$ after we already remove $A_2$,  there is at least one path $P$ from $S$ to $x$ that passes through $x$ but not through any element in $A_2$.
By the live-edge graph construction in the LT model, actually there is at most one path from $S$ to $v$.
Therefore, $P$ is the unique path from $S$ to $v$, and none of the nodes or edges in $A_2$ are on $P$ but $x$ is on $P$. 
Since $A_1 \subseteq A_2$, we know that after removing $A_1$, $P$ still exists but after removing $x$, $P$ no longer exists and there is no path from $S$ to $v$ anymore. 
This means that $v \in \Gamma(L, S) \setminus  \Gamma(L\setminus A_1 \setminus \{x\}, S)$ but $v\notin \Gamma(L, S) \setminus \Gamma(L\setminus A_1, S)$.
	Therefore, the lemma holds.
%
%
\end{proof}
The monotonicity and submodularity provide the theoretical basis for our efficient algorithm, to be presented in the next section. 



\section{Efficient Algorithms for AdvIM}\label{sec:scalablealgo}

The monotonicity and submodularity of the objective function enable the greedy approach for the maximization task. 
In this section, we aim to speed up the greedy approach by adapting the approach of reverse influence sampling (RIS)~\cite{BorgsBrautbarChayesLucier,tang14,tang15}, which provides both theoretical guarantee and efficiency.
We first provide a general adversarial attack algorithm framework $\AAIMM$ based on IMM~\cite{tang15} in Section \ref{sec:overviewIMM} (Algorithm \ref{alg:imm}).
$\AAIMM$ relies on valid reverse reachable (VRR) path sampling. 
Then in Section~\ref{sec:VRRsimulation}, we provide several concrete implementations of VRR path sampling. 

\subsection{Algorithm Framework $\AAIMM$} \label{sec:overviewIMM}

All efficient influence maximization algorithms such as IMM are based on the RIS approach, which generates a suitable number of reverse-reachable sets for influence estimation. 
In our case, we need to adapt the RIS approach for generating reverse-reachable paths in the LT model.
Let $S$ be the fixed seed set.
Let $L$ be a random live-edge graph generated from $G=(V,E,w)$ following the LT model.
Recall that each node selects at most one incoming edge as a live edge in $L$.
Thus, starting from a node $v\in V$, we may find at most one node $u_1$ such that $(u_1,v)\in L$.
Let $u_0 = v$.
In general, starting from $u_i\in V$, we may find at most one node $u_{i+1} \in V$ with $(u_{i+1},u_i) \in L$.
This process stops at some node $u_j$ when one of the following conditions hold:
(a) $u_j$ is a seed node, i.e. $u_j \in S$; 
(b) there is no edge $(u, u_j) \in E(L)$; or
(c) the path loops back, that is, the only edge $(u,u_j)\in E(L)$ satisfies $u \in \{u_0,\ldots, u_{j-1}\}$.
We call this process a reverse simulation from root $v$ in the LT model, and the path obtained from $u_j$ to $u_0$
a {\em reverse-reachable path (RR path) rooted at $v$} 
(under the live-edge graph $L$), denoted as $P_{L,v}$.
Note that if $L$ is a random live-edge graph, then 
$P_{L,v}$ is a random path, with randomness coming from $L$.
When we do not specify the root $v$, we define a {\em reverse-reachable path (RR path)} (under a random live-edge graph $L$) $P_L$ as a random $P_{L,v}$ with $v$ selected uniformly at random from the node set $V\setminus S$.
The reason we exclude the seed set $S$ is that the root $v$ selected in an RR path is to be used as a measure for the influence reduction of attack nodes or edges on the path, and since no seed node is selected in the attack set, no seed node will be counted for influence reduction.
This point will be made clearly and formally in the following lemma.
Sometimes we omit the subscript $L$ in $P_L$ when the context is clear.
Let $VE(P_L)$ be the joint set of nodes and edges of RR path $P_L$ excluding any seed node.

We say that an RR path $P_L$ is a valid RR path or VRR path which contains a seed node in $S$; otherwise $P_L$ is invalid.
Intuitively, for a VRR path $P_L$, the influence of $S$ can reach the root $v$ of $P_L$ through the path, and thus attacking any node or edge on the path would reduce the influence to $v$; but if $P_L$ is invalid, attaching a node or edge on $P_L$ will not reduce the influence to $v$ since anyway $v$ is not influenced by $S$.
For convenience, sometimes we also use the notation of $S$-conditioned RR path $P_L^S$, which is $P_L$ when $P_L$ is valid and $\emptyset$ when $P_L$ is invalid.
Let $\cP^S$ be the probability subspace of VRR paths.
Let $n^-=|V| - |S|$, $\I\{\}$ be the indicator function, and $\sigma^-(S) = \sigma(S) - |S|$.
The following lemma connects the influence reduction of an attack set $A$ with the VRR paths.

\begin{restatable}{lemma}{lemmaimm}\label{lem:VRRpathidentity}
For any given seed set $S$ and attack set $A$,
    \begin{align}
    \rho_S(A) & = n^- \cdot \E_L[\I\{A \cap VE(P_L^S) \ne \emptyset \} ] \nonumber \\
    & = \sigma^-(S) \cdot \E_L[\I\{A \cap VE(P_L) \ne \emptyset \} \mid S \cap VE(P_L) \ne \emptyset].  \nonumber 
    \end{align}
\end{restatable}
\begin{proof}
	By definition, we have
    \begin{align}
    & \rho_S(A) = \E_L [ | \{ v \in V\setminus S \mid v \in \Gamma(L,S) \wedge v\not\in \Gamma(L \setminus A, S)   \}   |  ]  \nonumber  \\
    & = \E_L \left[  n^- \cdot \E_{v \sim \cU(V\setminus S)} [\I\{v \in \Gamma(L,S) \wedge v\not\in \Gamma(L \setminus A, S)   \} ] \right]  \nonumber  \\
    & = n^-  \cdot  \E_{L, v \sim \cU(V\setminus S)}[\I\{A \cap P_L^S(v) \ne \emptyset \} ] \nonumber \\
    & = n^-   \cdot \Pr_{L, v \sim \cU(V\setminus S)}\{ S \cap P_L(v) \ne \emptyset\} \nonumber \\
    & \quad \cdot \E_{L, v \sim \cU(V\setminus S)}[\I\{A \cap P_L(v) \ne \emptyset \} \mid S \cap P_L(v) \ne \emptyset] \nonumber \\
    & = \sigma^-(S) \cdot \E_{L, v \sim \cU(V\setminus S)}[\I\{A \cap P_L(v) \ne \emptyset \} \mid S \cap P_L(v) \ne \emptyset] \label{eq:infidentity} \\
    & = \sigma^-(S) \cdot \E_{P \sim \cP^S} [\I  \{ A \cap P \ne \emptyset \}  ], \label{eq:subspace}
\end{align}
where $\cU(V\setminus S)$ denote the uniform distribution among all nodes in $V\setminus S$.
Eq.~\eqref{eq:infidentity} is due to the original RR set connection with influence spread~\cite{BorgsBrautbarChayesLucier,tang14,tang15}, which shows that $\Pr_{L, v \sim \cU(V\setminus S)}\{ S \cap P_L(v) \ne \emptyset\} = \sigma^-(S) / n^- $.
Eq.~\eqref{eq:subspace} follows from the subspace definition of $\cP^S$.
\end{proof}

The above property implies that we can sample enough RR path from the original space or VRR path from subspace $\cP^S$ to accurately estimate the influence reduction of $A$.
More importantly, by Lemma~\ref{lem:VRRpathidentity} the optimal attack set can be found by seeking the optimal set of nodes and edges that intersect with (a.k.a. cover) the most number of VRR paths, which is a max-cover problem.
Therefore, following the RIS approach, we turn the influence reduction maximization problem into a max-cover problem.
We use the IMM algorithm~\cite{tang15} as the template, but other RIS algorithms follow the similar structure.
Our algorithm $\AAIMM$ contains two phases, estimating the number of VRR paths needed and greedy selection via max-cover, 
as shown in Algorithm \ref{alg:imm}.
The two main parameters $\lambda'$ and $\lambda^*(\ell)$ used in the algorithm are given below:
\begin{align}
& \lambda' \leftarrow (2+\frac{2}{3}\varepsilon') \left(\ln{\left(\binom{n^-}{q_N}\binom{m}{q_E}\right)}+\ell \ln{n^-}+\ln{\log_2{n^-}}\right)\cdot n^-
\label{eq:lambdaprime} \\
& \lambda^*(\ell) \leftarrow 2 n \cdot (1/2 \cdot \alpha + \beta)^2 \cdot \varepsilon^{-2} 
	\label{eq:lambdastar} \\
& 	\alpha \leftarrow \sqrt{\ell \ln{n} + \ln{2}};
\beta \leftarrow \sqrt{1/2 \cdot \left(\ln{\left(\binom{n^-}{q_N}\binom{m}{q_E}\right)}+\alpha^2\right)}. \nonumber
\end{align}

\begin{algorithm}[t] 
	\caption{{\AAIMM}: Adversarial Attacks IMM 
	} \label{alg:imm}
	\KwIn{Graph $G=(V,E,w)$, seed set $S$, budgets $q_N, q_E$, 
		accuracy parameters $(\varepsilon, \ell)$}
	
	\tcp{Phase 1: Estimate $\theta$, the number of VRR paths needed, and generate these VRR paths}
	$\newalg{\cR \leftarrow \emptyset}$;  $LB \leftarrow 1$; $\varepsilon' \leftarrow \sqrt{2}\varepsilon$;
	using binary search to find a $\gamma$ such that
	$\lceil\lambda^*(\ell) \rceil /n^{\ell+\gamma} \le 1/n^\ell$ 
	$\newalg{\ell \leftarrow \ell + \gamma + \ln 2 / \ln{n^-}} $\;
	
	\For{$i = 1 \; to \; \log_2{(\newalg{n} - 1)}$ \label{line:bimfor1}}{
		$x_i \leftarrow n^- /2^i$\; \label{line:bimassignx}
		$\theta_i \leftarrow \lambda' \cdot \varepsilon'^{-2} /x_i$; \tcp{$\lambda'$ is defined in Eq.~\eqref{eq:lambdaprime}} \label{line:bimgenRRset1b} 
		\While{$\newalg{|\cR|}  < \theta_i$}{
			Sample a VRR path $P$ from subspace $\cP^S$, and insert it into $\cR$\;
			\label{line:bimgsample1}
		}\label{line:bimgenRRset1e} 
		\newalg{$A_i \leftarrow \NodeEdgeSelection(\cR, q_N, q_E)$}\; \label{line:bimnodeselect1}
		\If{$n^- \cdot \newalg{F^S_{\cR}(A_i)} \geq (1 + \varepsilon') \cdot x_i$}{
			\label{line:bimestimate1}
			$LB \leftarrow n^- \cdot F^S_{\cR}(A_i)/(1 + \varepsilon')$\; 
			\label{line:bimestimate2}
			{\bf break}\;  \label{line:endcheck}
		} 
	}
	$\theta \leftarrow \lambda^*(\ell) / LB$; \tcp{$\lambda^*(\ell)$ is defined in Eq.~\eqref{eq:lambdastar}}  \label{line:settheta}
	\While{$\newalg{|\cR|} \leq \theta$}{ \label{line:bimcheckLB}
		Sample a VRR path $P$ from subspace $\cP^S$, and insert it into $\cR$\;
		\label{line:bimnodeselect2}		
	}
	\tcp{Phase 2: select attack nodes and edges from the generated VRR paths}
	$A \leftarrow \NodeEdgeSelection(\cR, q_N, q_E)$\;\label{line:bimgsample2}	
	\bf{return} an attack set $A$
\end{algorithm}

In Phase 1, we generate $\theta$ valid VRR paths $\cR$, where $\theta$ is computed to guarantee the approximation with high probability.
In Phase 2, we use the greedy algorithm to find the $q_N$ nodes and $q_E$ edges that cover the most number of VRR paths. 
This greedy algorithm is similar to prior algorithms, and thus we omit it here.
Phase 1 follows the IMM structure to estimate a lower bound of the optimal value, which is used to determine the number of VRR paths needed.
The main difference is that we need to sample valid RR paths, not the simple RR sets as before. 
This part will be discussed separately in the next section.
Moreover, our solution space now is $\binom{n^-}{q_N}\binom{m}{q_E}$, and thus we replace the $\binom{n}{k}$ in the original IMM algorithm.
%
Let $A^*$ be the optimal solution of the AdvIM problem, and $\OPT = \rho_S(A^*)$.

\begin{restatable}{lemma}{sampling} \label{lem:sampling}
	For every $\varepsilon > 0$ and $\ell > 0$, to guarantee the approximation ratio with probability at least
	$1-\frac{1}{n^\ell}$, the number of VRR paths needed by \alg{\AAIMM} is $O\left( \frac{(q_N\log{n^-}+q_E\log{m}+\ell\log{n^-})\cdot \sigma^-(S)}{\OPT \cdot \varepsilon^2} \right)$.
\end{restatable}
\begin{proof}
	When working on the subspace $\cP^S$ in \alg{\AAIMM}, we are working on the objective function $n^- \cdot \E_{P\sim \cP^S}[\I\{A \cap VE(P) \ne \emptyset \}]$
		(e.g. see lines~\ref{line:bimassignx}, \ref{line:bimestimate1}, and \ref{line:bimestimate2}).
	Thus by Lemma~\ref{lem:VRRpathidentity}, the real objective function $\rho_S(A)$ is only a fraction $\frac{\sigma^-(S)}{n^-}$ of the new objective function.
    Let $\OPT'$ be the optimal value of the new objective function.
    Applying the analysis of IMM~\cite{tang15}, we know that the number of VRR path samples that we need in the subspace $\cP^S$ is
    \begin{align}
        O\left( \frac{(q_N\log{n^-}+q_E\log{m}+\ell\log{n^-})\cdot n^-}{\OPT'\cdot \varepsilon^2} \right). \label{eq:subRRPsamples}
    \end{align}
    Because $\OPT = \frac{\sigma^-(S)}{n^-} \cdot \OPT'$, the above formula is changed to 
    \begin{align}
    O\left( \frac{(q_N\log{n^-}+q_E\log{m}+\ell\log{n^-})\cdot \sigma^-(S)}{\OPT \cdot \varepsilon^2} \right), \label{eq:subRRPsamples2}
\end{align}
\end{proof}

The special case of adversarial attacks on influence maximization is attacking nodes only or edge only, i.e, $q_E=0$ or $q_N=0$. Then, the greedy algorithms of the special cases can achieve at least $(1 - 1/e - \varepsilon)$ of the optimal performance.
However, our greedy algorithms for the AdvIM attacks both nodes and edges together.
The idea of our greedy approach is that at every greedy step, it searches all nodes and edges in the candidate space $\cC$ and picks the one having the maximum marginal influence deduction.
If the budget for node or edge exhausts, then the remaining nodes or edges are removed from $\cC$.
Note that as $\cC$ contains nodes and edges assigned to different partitions, \alg{\AAIMM} selects nodes or edges crossing partitions.
This falls into a greedy algorithm subject to a partition matroid constraint, which is defined below.

Given a set $U$ partitioned into disjoint sets $U_1, \dots, U_n$ and $\mathcal{I}=\{X\subseteq U\colon |X\cap U_i|\le k_i, \forall i\in[n]\}$, $(U,\mathcal{I})$ is called a {\em partition matroid}.
Thus, the node and edge space $\cA$ with the constraint of AdvIM, namely $(\cA, \{A\colon|A_N|\le q_N,|A_E|\le q_E\})$, is a partition matroid.
This indicates that AdvIM is an instance of submodular maximization under partition matroid, which can be solved by a greedy algorithm with
	$1/2$-approximation guarantee \cite{fisher1978analysis}.
Using this result, we can obtain the result for our \alg{\AAIMM} algorithm.

\begin{restatable}{theorem}{theoremimm} \label{thm:imm}
	For every $\varepsilon > 0$ and $\ell > 0$, with probability at least
	$1-\frac{1}{n^\ell}$, the output $A^o$ of the \alg{\AAIMM} algorithm framework satisfies $\rho_S(A^o) \geq \left(\frac{1}{2} - \varepsilon\right) \rho_S(A^*)$.
	In this case, the expected running time for \alg{\AAIMM} is $O\left( \frac{(q_N\log{n^-}+q_E\log{m}+\ell\log{n^-})\cdot \sigma^-(S)}{\OPT\cdot \varepsilon^2} \cdot \ERPN \right)$,
    where $\ERPN$ is the mean time of generating a VRR path.
\end{restatable}
\begin{proof}
Following \cite{fisher1978analysis}, our $\NodeEdgeSelection$ procedure finds an attack set that covers at least $1/2$ of all the VRR paths generated.
Then following the standard analysis of IMM, our \alg{\AAIMM} algorithm provides $1/2-\varepsilon$ approximation with probability at least
	$1-\frac{1}{n^\ell}$.
The expected running time follows Lemma~\ref{lem:sampling}.
\end{proof}
Theorem~\ref{thm:imm} summarizes the theoretical guarantee of the approximation ratio and the running time of our algorithm framework \alg{\AAIMM}.
Because we are attacking both nodes and edges with separate budgets, our approximation guarantee is $1/2 - \varepsilon$.
If we only attack nodes (i.e. $q_E=0$) or only attack edges (i.e. $q_N=0$), then we could achieve the approximation ratio of $1-1/e-\varepsilon$.
The running time result includes $\OPT$.
Although $\OPT$ cannot be obtained in general, it still provides a precise idea on the running time, and is useful to compare different implementations.
In the next section, we will discuss concrete implementations of the VRR path sampling method, and in some cases we will fully realize the time complexity.

\subsection{VRR Path Sampling} \label{sec:VRRsimulation}

VRR path sampling is the key new component to fully realize our \alg{\AAIMM} algorithm.
Here, we discuss three methods and their guarantees, and empirically evaluate these methods in our experiments.

%
%

\paragraph{Naive VRR Path Sampling.}
The first implementation is naively generate a RR path $P$ starting from a random root $v \in V \setminus S$, and if $P$ is valid then return it; otherwise
	regenerate a new path (see Algorithm~\ref{alg:nvrrp}).
It is easy to see that to generate one VRR path, we need to generate a number of RR paths.
Let $\ERP$ be the expected running time of generating one RR path.
By a simple argument based on live-edge graphs, we can get that on average we need to generate $n^- / \sigma^-(S)$ RR paths to obtain one VRR path that reaches the seed set $S$.
This means $\ERPN = \ERP \cdot n^- / \sigma^-(S)$.

\begin{algorithm}[h] 
	\caption{{\NVRRP}: Naive VRR Path Sampling} \label{alg:nvrrp}
	\KwIn{Graph $G$, seed set $S$}
	
	\Repeat{$P \cap S \neq \emptyset$}{
		Randomly select a root $v \in V \setminus S$, and generate the reverse-reachable path $P$ rooted at $v$\;
	}
	\bf{return} a VRR path $P$.
\end{algorithm}

\begin{restatable}{theorem}{naiveimm} \label{thm:naiveimm}
	Naive VRR path sampling (Algorithm~\ref{alg:nvrrp}) correctly samples a VRR path.
	The expected running time of $\AAIMM$ with naive VRR path sampling (Algorithm~\ref{alg:nvrrp}) is 
    \begin{align}
    O\left( \frac{(q_N\log{n^-}+q_E\log{m}+\ell\log{n^-})\cdot n^-}{\OPT\cdot \varepsilon^2} \cdot \ERP \right). \label{eq:naiveimm}
    \end{align}
\end{restatable}
\begin{proof}
It is obvious that the naive VRR path sampling will return a VRR path according to distribution $\cP^S$.
We just need to prove that $\ERPN = \ERP \cdot n^- / \sigma^-(S)$, and the rest follows Theorem~\ref{thm:imm}.
For each live-edge graph $L$, if we randomly select a root $v \in V\setminus S$, then with probability $(|\Gamma(L,S)|-|S|)/n^-$, $v$ is reachable from $S$ in $L$, which means the RR path from $v$ will intersect with $S$ on this live-edge graph $L$.
Taking expectation over $L$, we know that the probability of an RR path is valid is $\E_L [(|\Gamma(L,S)|-|S|)/n^- ] = \sigma^-(S)/n^-$.
Therefore, on average we need to generate $ n^- / \sigma^-(S)$ RR paths to get one VRR path, which means $\ERPN = \ERP \cdot n^- / \sigma^-(S)$.
%
\end{proof}

\paragraph{Forward-Backward VRR Path Sampling.}
To avoid wasting RR path samplings as in the naive method, we can first do a forward simulation from $S$ to generate a forward forest, recording the nodes and edges that a forward simulation from $S$ will pass.
Then, when randomly selecting a root $v$, we restrict the selection to be among the nodes touched by the forward simulation.
Finally, the VRR path is the path from $S$ to $v$ recorded in the forward forest.
This is the forward-backward sampling method given in Algorithm~\ref{alg:fbvrrp}.
Let $\EFF(S)$ be the meantime of generating a forward forest from seed set $S$.
Then we have the following result on this method.

\begin{algorithm}[!h] 
	\caption{{\FBVRRP}: Forward-Backward VRR Path Sampling} \label{alg:fbvrrp}
	\KwIn{Graph $G$, seed set $S$}
	
		Initialize an empty forest $F$\;
		Forward propagating a new forest $F$ by using LT model with $S$ and $G$\;
		Randomly select a node $v \in F \setminus S$, and set path $P$ to be the one from $S$ to $v$ in the forest $F$\;
		$\cR^o \leftarrow \cR^o \cup P$;
	\bf{return} a VRR path $P$.
\end{algorithm}

\begin{restatable}{theorem}{fbimm} \label{thm:fbimm}
	Forward-backward VRR path sampling (Algorithm~\ref{alg:fbvrrp}) correctly samples a VRR path.
	The expected running time of $\AAIMM$ with forward-backward VRR path sampling (Algorithm~\ref{alg:fbvrrp}) is 
    \begin{align}
    O\left( \frac{(q_N\log{n^-}+q_E\log{m}+\ell\log{n^-})\cdot \sigma^-(S)}{\OPT\cdot \varepsilon^2} \cdot \EFF(S) \right). \label{eq:fbimm}
    \end{align}
\end{restatable}
\begin{proof}
For any fixed live-edge graph $L$, conditioned on sampling a VRR path, random sampling a root $v$ is equivalent of sampling from $\Gamma(L,S) \setminus S$
	uniformly at random, which is exactly from the forward forest generated by forward simulation from $S$.
Thus, the forward-backward sampling method is correct.
\end{proof}

Compared with naive sampling, the forward-backward sampling saves those sampling of invalid RR paths.
However, it needs to generate a complete forward forest first, which is more expensive than generating one reverse path.
Therefore, there is a tradeoff between the forward-backward method and the naive method, and the tradeoff is exactly quantified by their running time results: If $\EFF(S) / \ERP < n^- / \sigma^-(S)$, then the forward-backward method is faster; otherwise, the naive method is faster.
Therefore, which one is better will depend on the actual graph instance.

Since the forward-backward method generates a forward forest first, one may be tempted to sample more VRR paths from this forest, but it will generate correlations among these VRR paths.
Another attempt of generating a number of forward forests to record the frequencies of node appearances, and then use these frequencies to guide the RR path sampling would also deviate from the subspace probability distribution $\cP^S$, and thus these methods would not provide theoretical guarantee of the overall correctness of \alg{\AAIMM}.

\paragraph{Reverse-Reachable Simulation with DAG}
The naive method above wastes many invalid RR path samplings, while the forward-backward method wastes many branches in the forward forest.
Thus, we desire a method that could do a simple VRR path sampling without such waste.
For directed-acyclic graphs (DAGs), we do discover such a VRR path sampling method by re-weighting edges.


Suppose $G$ is a directed acyclic graph. As shown in \cite{ChenYZ10}, in a DAG, influence spread of a seed set can be computed in linear time.
Let $\ap_v(S)$ be the probability of $v$ being activated when $S$ is the seed set.
Let $N^-(v)$ be the set of $v$'s in-neighbors. Then in a DAG $G$, we have $\ap_v(S) = \sum_{u\in N^-(v)} \ap_u(S) \cdot w(u,v)$.
The above computation can be carried out with one traversal of the DAG in linear time from the seed set $S$ following any topological sort order.
For a DAG $G$, we propose the following sampling of a VRR path in $\cP^S$ as \alg{\DAGVRRP} in Algorithm \ref{alg:vdagp}.

\begin{algorithm}[h] 
	\caption{{\DAGVRRP}: DAG VRR Path Sampling} \label{alg:vdagp}
	\KwIn{Graph $G$, seed set $S$}
	
	    Randomly sample root $v \in V \setminus S$ with probability proportional to $ap_v(S)$, i.e.  with probability $\frac{\ap_v(S)}{\sum_{u\in V\setminus S} \ap_u(S)}= \frac{\ap_v(S)}{\sigma^-(S)}$\;
	    $u_0 \leftarrow v; P \leftarrow \{u_0\}$; $i \leftarrow 0$\;
    	    \Repeat{$u_{i} \in S$ }{
    	        Sample $u_{i+1}\in N^-(u_i)$ with probability proportional to $\ap_{u_{i+1}}(S) \cdot w(u_{i+1},u_i)$, i.e. 
    	        $ \frac{\ap_{u_{i+1}}(S) \cdot w(u_{i+1},u_i)}{\sum_{u\in N^-(u_i)} \ap_{u}(S) \cdot w(u,u_i)}=
    		    \frac{\ap_{u_{i+1}}(S) \cdot w(u_{i+1},u_i)}{\ap_{u_i}(S)}$\;
    		    Add node $u_{i+1}$ and edge $(u_{i+1}, u_i)$ into path $P$\;
    		    $i \leftarrow i+1$\;
    	    }
	\bf{return} a VRR path $P$.
\end{algorithm}



\begin{restatable}{lemma}{dag} \label{lem:dag}
	If the graph $G$ is a DAG, then any path sampled by Algorithm \alg{\DAGVRRP} follows the subspace distribution $\cP^S$.
\end{restatable}
\begin{proof}
	First, notice that Algorithm \alg{\DAGVRRP} re-weights the incoming edges $(u,v)$ of every node $v$ according to the activation probability $\ap_u(S)$.
	Therefore any node $u$ that cannot be activated by $S$ will cause $(u,v)$ to have zero weight, and thus will not be sampled in the reverse sampling process.
	Therefore, the reverse sampling process will always sample toward the seed set, and since the graph has no cycle, it will always end at a seed node in $S$.
	This means that the output of \alg{\DAGVRRP} is always a VRR path.
	We just need to show that it follows the subspace distribution $\cP^S$.
	
	To show that its distribution is the same as $\cP^S$, 
	all we need to show is that if we have two paths $P_1$ and $P_2$ that both start from a seed node in $S$, then the ratio of the probabilities
	of generating these two paths by the above procedure is the same as the ratio in the original RR path space. 
	Let $P_1 = (u_{s}, u_{s-1}, \ldots, u_0)$ and $P_2 = (z_t, z_{t-1}, \ldots, z_0)$, with $u_0, z_0 \in V\setminus S$ and $u_s, z_t \in S$.
	Let $\pi(P)$ be the probability of sampling the RR path $P$ in the original space.
	Then we have $\pi(P_1) = \frac{1}{n^-} \prod_{i=1}^s w(u_{i}, u_{i-1})$, and $\pi(P_2) = \frac{1}{n^-} \prod_{j=1}^t w(z_{j}, z_{j-1})$.
	So the ratio is
	\begin{align}
	\frac{\pi(P_1)}{\pi(P_2)} = \frac{\prod_{i=1}^s w(u_{i}, u_{i-1})}{\prod_{j=1}^t w(z_{j}, z_{j-1})}.
	\end{align}
	
	Now let $\pi'(P)$ be the probability of generating path $P$ by \alg{\DAGVRRP}. 
	Then we have 
	\begin{align}
	\pi'(P_1) = \frac{\ap_{u_0}(S)}{\sigma^-(S)} \prod_{i=1}^s \frac{\ap_{u_{i}}(S) \cdot w(u_{i},u_{i-1})}{\ap_{u_{i-1}}(S)} 
	= \frac{\prod_{i=1}^s w(u_{i}, u_{i-1})}{\sigma^-(S)},
	\end{align}
	Where the above derivation also uses the fact that $u_s\in S$ and thus $\ap_{u_s}(S)=1$.
	Similarly, we have
	\begin{align}
	\pi'(P_2) = \frac{\ap_{z_0}(S)}{\sigma^-(S)} \prod_{j=1}^t \frac{\ap_{z_{j}}(S) \cdot w(z_{j},z_{j-1})}{\ap_{z_{j-1}}(S)} 
	= \frac{\prod_{j=1}^t w(z_{j}, z_{j-1})}{\sigma^-(S)}.
	\end{align}

    Therefore, clearly $\pi'(P_1)/\pi'(P_2) = \pi(P_1)/\pi(P_2)$, and the above procedure correctly generates a VRR path from $\cP^S$.
\end{proof}
	
Therefore we can use \alg{\DAGVRRP} sample from $\cP^S$.
Now we just need to analyze its running time $\ERPN$.
The generation is still similar to the LT reverse simulation.
For any $u\not\in S$, let $\tau(u)$ be the time needed to do a reverse simulation step from $u$. 
For the LT model, a simple binary search implementation takes $\tau(u) = O(\log_2 d_u^-)$ time, where $d_u^-$ is the indegree of $u$.
For an RR path $P$, let $\omega(P) = \sum_{u\in V(P)\setminus S} \tau(u)$ be the total time needed to generate path $P$.
Let $\tau = \sum_{u\in V \setminus S} \tau(u)$.

\begin{restatable}{lemma}{erpn} \label{lem:ERPN}
	Let $\tilde{v}$ be a randomly sampled node from $V\setminus S$, 
	with sample probability proportional to $\tau(v)$. 
	Let $P$ be a random VRR path generated by \alg{\DAGVRRP}, then we have $\ERPN =  \E_{P\sim \cP^S}[\omega(P)] = \frac{\tau}{\sigma^-(S)} \cdot \E_{\tilde{v}}[\rho_S(\{\tilde{v}\})]$.
\end{restatable}
\begin{proof}
	For a fixed RR path $P$, let $p(P)$ be the probability $\tilde{v} \in V(P)$.
	Then it is clear that
	\begin{equation*}
	p(P) = \E_{\tilde{v}}[\I\{\tilde{v}\in V(P)\}] = \frac{\omega(P)}{\tau}.
	\end{equation*}
	
        Let $P$ be a random VRR path generated by Algorithm \ref{alg:vdagp}.
	Then  
	\begin{align*}
	& \E_{P\sim \cP^S}[\omega(P)] = \tau \cdot \E_{P\sim \cP^S}[p(P)]
	= \tau \cdot \E_{P\sim \cP^S}[\E_{\tilde{v}}[\I\{\tilde{v}\in V(P)\}]] \\
	&= \tau \cdot \E_{\tilde{v}}[\E_{P\sim \cP^S}[\I\{\tilde{v}\in V(P)\}]] = \frac{\tau}{\sigma^-(S)} \cdot \E_{\tilde{v}}[\rho_S(\{\tilde{v}\})],
	\end{align*}
	where the last equality is by Lemma~\ref{lem:VRRpathidentity}.
\end{proof}

Finally, applying the above $\ERPN$ result to Theorem~\ref{thm:imm}, we can obtain the following

\begin{restatable}{theorem}{dagimm} \label{thm:dagimm}
	Algorithm \alg{\DAGVRRP} correctly samples a VRR path.
	The expected running time of $\AAIMM$ with \alg{\DAGVRRP} sampling is $O\left( \frac{(q_N\log{n^-}+q_E\log{m}+\ell\log{n^-})\cdot \tau \cdot \E_{\tilde{v}}[\rho_S(\{\tilde{v}\})]}{\OPT\cdot \varepsilon^2} \right)$,
    where $\tau = O(\sum_{u\in V\setminus S} \log_2 d^-_u) = O(n\log n)$. 
\end{restatable}
Notice that when $q_N\ge 1$, we have $\E_{\tilde{v}}[\rho_S(\{\tilde{v}\})] \le \OPT$.
Therefore, in this case we will have a near-linear-time algorithm, just as the original influence maximization algorithm IMM.

The above result relies on that $G$ is a DAG.
When the original graph is not  DAG, we can transform the graph to a DAG, similar to the DAG generation algorithm in 
	the LDAG algorithm (Algorithm 3 in \cite{chen2011influence}).
The difference is that in \cite{chen2011influence} it is generating a DAG to approximate the influence towards a root $v$.
Instead, in our case we want to generate a DAG that approximates the influence from seed set $S$ to other nodes.
But the approach is similar, and we can efficiently implement this DAG generation process by a Dijkstra short-path-like algorithm just as in \cite{chen2011influence}.


\begin{figure*}[b]
	\centering
	\subfloat[DBLP ($k=300$)]{\includegraphics[width=1.24in]{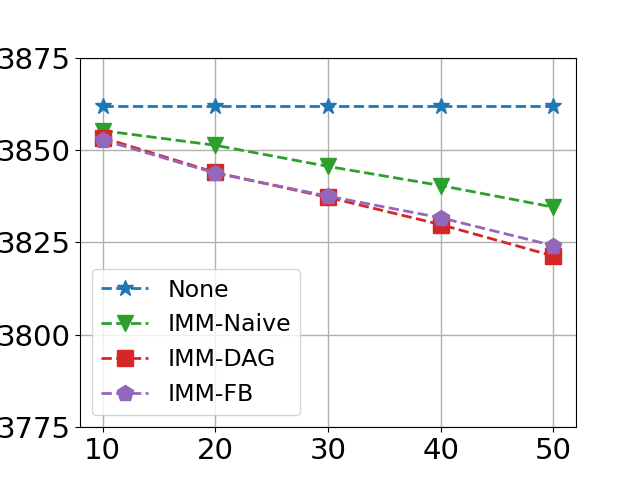}}\ \ \
	\subfloat[NetHEPT ($k=200$)]{\includegraphics[width=1.24in]{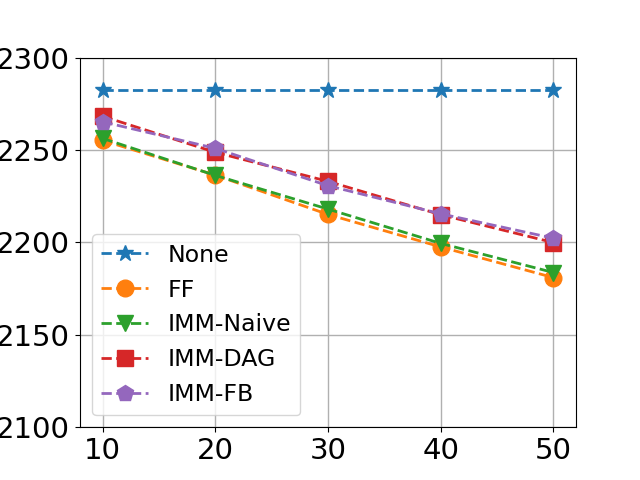}}\ \ \
	\subfloat[Flixter ($k=100$)]{\includegraphics[width=1.24in]{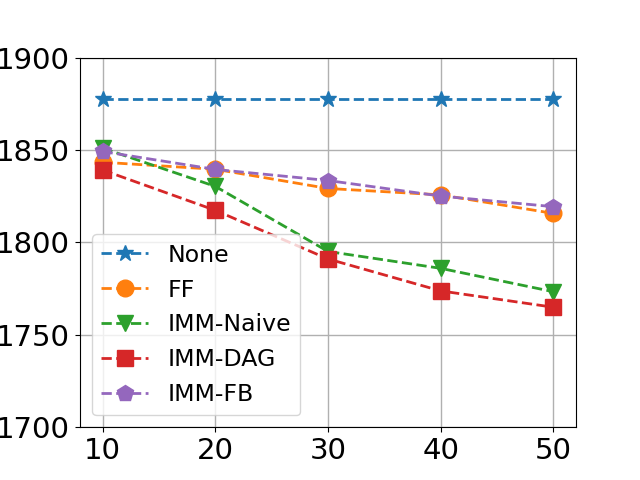}}\ \ \
	\subfloat[DM ($k=50$)]{\includegraphics[width=1.24in]{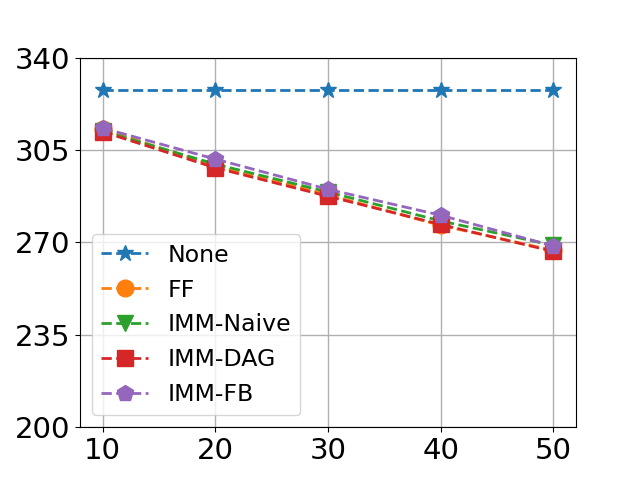}}\\
	\subfloat[DBLP ($k=300$, $q_E=10$)]{\includegraphics[width=1.24in]{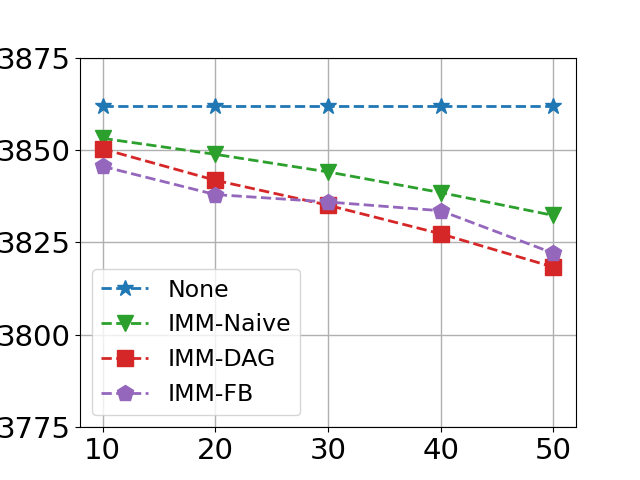}}\ \ \
	\subfloat[NetHEPT ($k=200$, $q_E=10$)]{\includegraphics[width=1.24in]{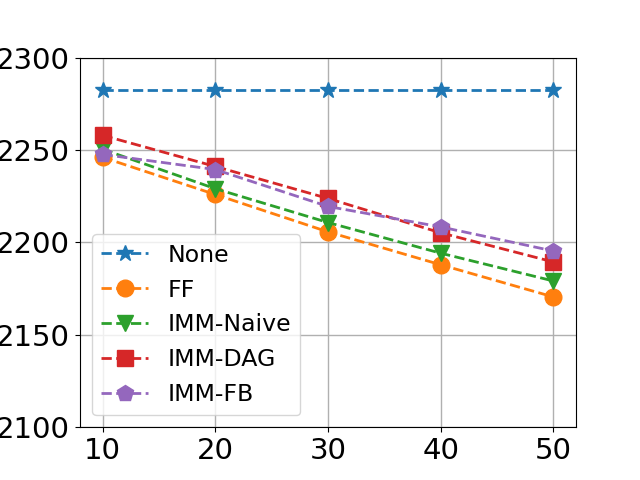}}\ \ \
	\subfloat[Flixter ($k=100$, $q_E=10$)]{\includegraphics[width=1.24in]{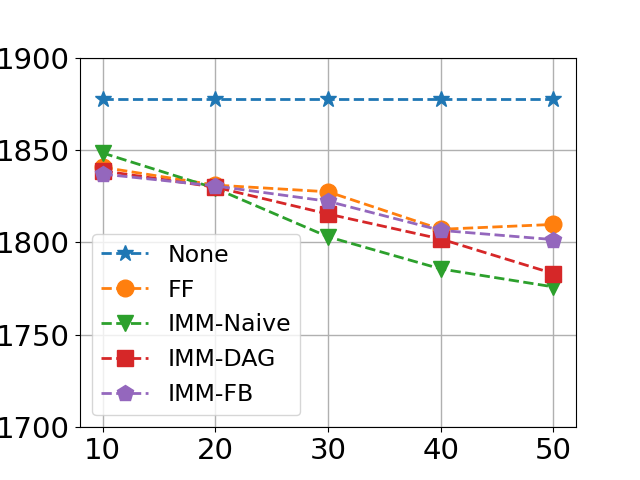}}\ \ \
	\subfloat[DM ($k=50$, $q_E=10$)]{\includegraphics[width=1.24in]{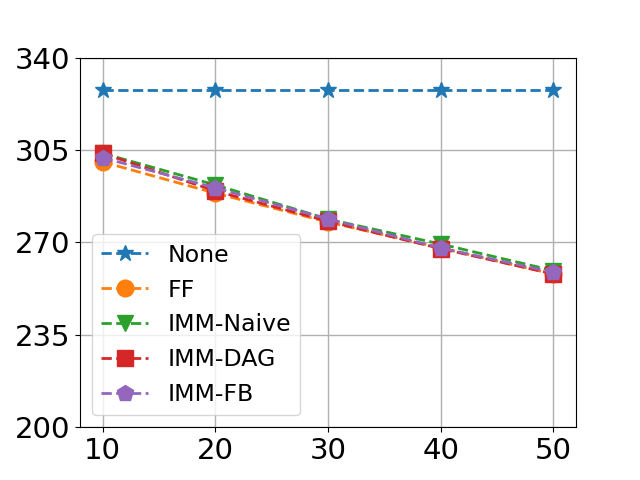}}\\
	\subfloat[DBLP ($k=300$)]{\includegraphics[width=1.24in]{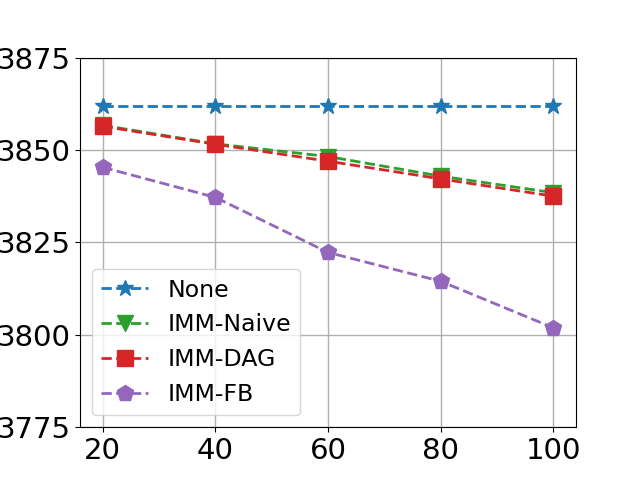}}\ \ \
	\subfloat[NetHEPT ($k=200$]{\includegraphics[width=1.24in]{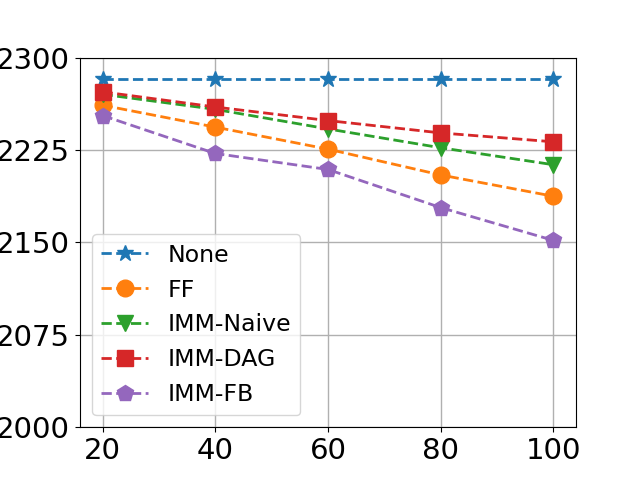}}\ \ \
	\subfloat[Flixter ($k=100$)]{\includegraphics[width=1.24in]{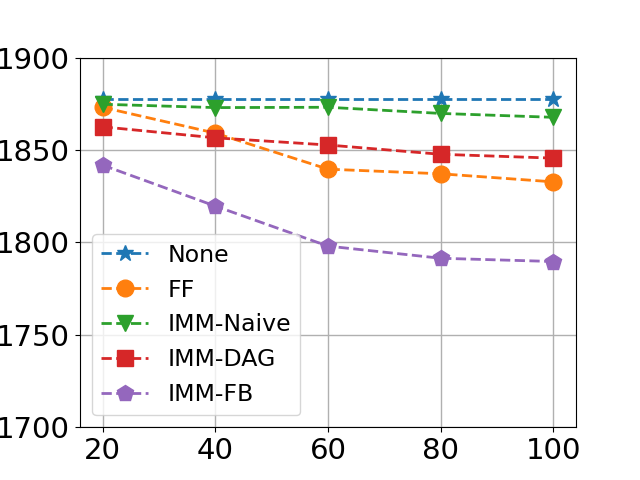}}\ \ \
	\subfloat[DM ($k=50$)]{\includegraphics[width=1.24in]{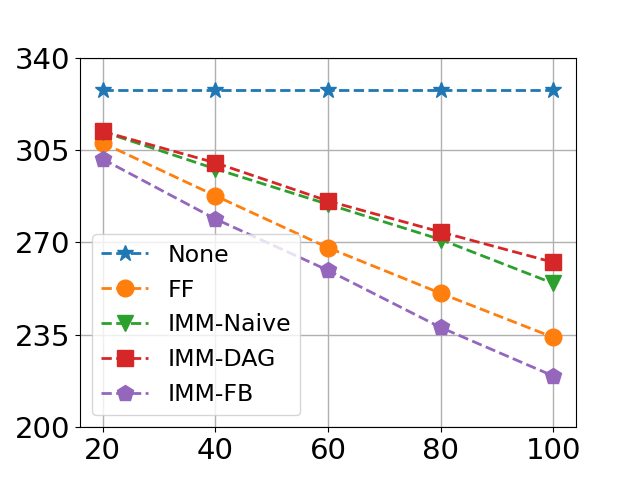}}\\
	\subfloat[DBLP ($k=300$, $q_N=20$)]{\includegraphics[width=1.24in]{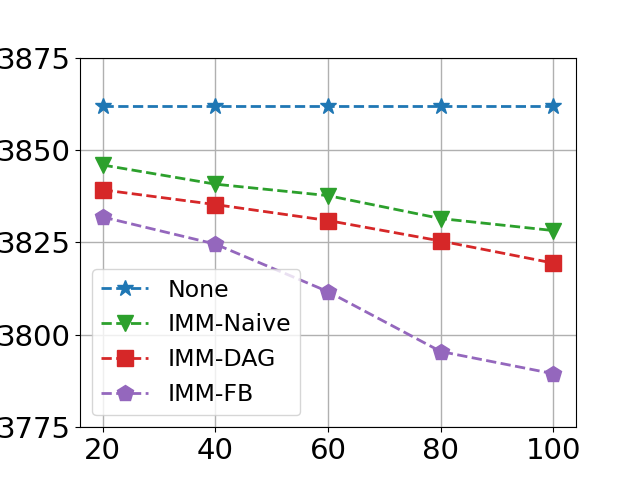}}\ \ \
	\subfloat[NetHEPT ($k=200$, $q_N=20$)]{\includegraphics[width=1.24in]{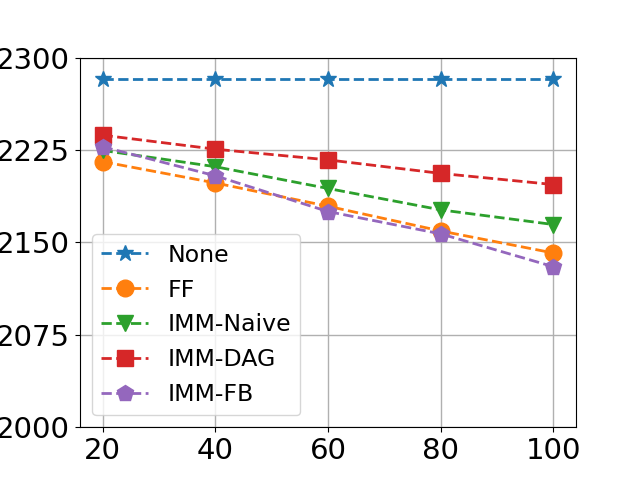}}\ \ \
	\subfloat[Flixter ($k=100$, $q_N=20$)]{\includegraphics[width=1.24in]{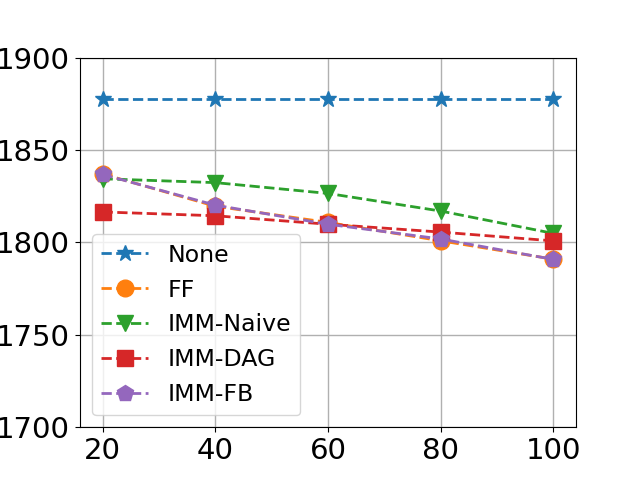}}\ \ \
	\subfloat[DM ($k=50$, $q_N=20$)]{\includegraphics[width=1.24in]{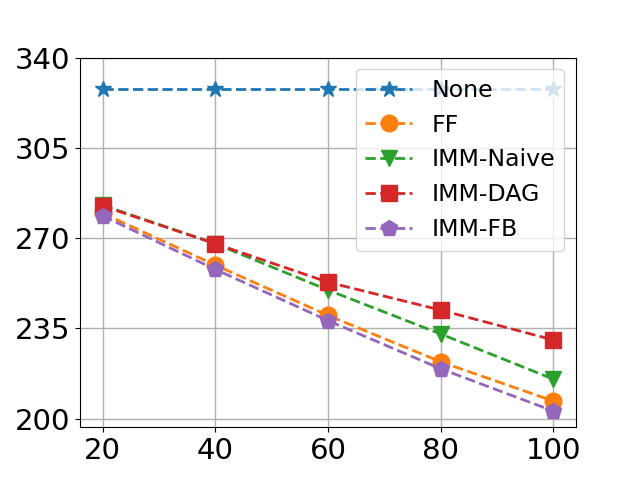}}\\
    \caption{Evaluation budget $q_N$ and $q_E$ on four datasets. We fixed $q_E$ for (a)-(h) and fixed $q_N$ for (i)-(p). 
    Once we fix the node budget, then the x-axis is the edge budget, and vice versa. The y-axis is the influence spread.} \label{fig:budget}
\end{figure*}

\section{Experimental Evaluation} \label{sec:experiment}


\subsection{Data and Algorithms}

\noindent\textbf{DBLP}. The DBLP dataset \cite{tang2009social} is a network of data mining, where every node is an author and every edge means the two authors collaborated on a paper.
The original DBLP is a graph containing \num{654628} nodes and \num{3980318} directed edges.
However, due to the limited memory of our computer (16GB memory, 1.4GHz quad-core Intel CPU), it hardly launches the whole test on such graph. So we randomly sample \num{100000} nodes and their \num{747178} directed edges from the original graph, which is still the largest size of graph among all four datasets.

\noindent\textbf{NetHEPT}. 
The NetHEPT dataset \cite{chen08} is extensively used in many influence maximization studies. 
It is an academic collaboration network from the ``High Energy Physics Theory'' section of arXiv from 1991 to 2003,
where nodes represent the authors and each edge represents one paper co-authored by two nodes. 
We clean the dataset by removing duplicated edges and obtain a directed graph $G = (V,E)$,
$|V|$ = \num{15233}, $|E|$ = \num{62774} (directed edges). 

\noindent\textbf{Flixster}.
The Flixster dataset \cite{barbieri2012topic} is a network of American social movie discovery services. 
To transform the dataset into a weighted graph, 
each user is represented by a node, 
and a directed edge from node $u$ to $v$ is formed if $v$ rates one movie shortly 
after $u$ does so on the same movie.
The Flixster graph contains \num{29357} nodes and \num{212614} directed edges.

\noindent\textbf{DM}. The DM dataset \cite{tang2009social} is a network of data mining researchers extracted from the ArnetMiner archive (aminer.org), where nodes present the researchers and each edge is the paper coauthorship between any two researchers.
DM is the small size dataset here, which only includes \num{679} nodes and \num{3374} directed edges.

\subsection{Algorithms}

We test all four algorithms proposed in the experiment, for the AdvIM task with different settings.
Some further details of each algorithm are explained below.

\alg{\AAFF}. This is the forward forest greedy algorithm. 
The number of forward-forest simulations of \alg{\AAFF} is the same as the reverse-reachable approaches in Theorem \ref{thm:naiveimm}. However, instead of using VRR paths, \alg{\AAFF} simulate a forward forest, i.e., a set of trees, in each propagation, which is required to occupy a huge computer memory to ensure theoretical guarantee.
To make it practical, we follow the standard practice in the literature and set the number of simulations as $10000$~\cite{kempe03,chen08,WCW12}. 
\OnlyInShort{The pseudo code and full analysis of \alg{\AAFF} is given in the appendix of the full version~\cite{arxivfull}.}

\alg{\AAIMMNaive}. This is the AAIMM algorithm with naive VRR path simulation, as given in Algorithm~\ref{alg:imm} and Algorithm~\ref{alg:nvrrp}. 
\alg{\AAIMMNaive} needs to generate plenty of RR paths for enough VRR paths since most naive RR paths can not touch the seed set $S$.
Compared to \alg{\AAFF}, \alg{\AAIMMNaive} uses much less computer memory cost to ensure the theoretical guarantee due to the RIS approach.

\alg{\AAIMMFB}.
This is the AAIMM algorithm with forward-backward VRR path simulation, as given in Algorithm~\ref{alg:imm} and Algorithm~\ref{alg:fbvrrp}.
Unlike \alg{\AAIMMNaive}, no path will be wasted in \alg{\AAIMMFB}, since all paths here are VRR paths that are randomly selected from the forward forests.
To fair compare the running time with \alg{\AAFF}, We choose to sample the same number of simulations of the forward forests.
The results show that compared with \alg{\AAFF}, \alg{\AAIMMFB} can save more computer memory and computation power.

\alg{\AAIMMDAG}.
This is the DAG-based AAIMM algorithm, as given in Algorithm~\ref{alg:imm} and Algorithm~\ref{alg:vdagp}.
Before simulation the VRR paths from DAG, we need first create a DAG as same as \cite{chen2011influence}.
Compared to previous RIS approaches, \alg{\AAIMMDAG} is the fastest approach for VRR path sampling.

We also use $10000$ Monte Carlo simulations for influence spread estimation after adversarial attacks for all the above approaches. 

\subsection{Result}

We test all four algorithms proposed in the experiment, for the Adv-IM with $k$ seeds budget, such as \alg{\AAFF}, \alg{\AAIMMNaive}, \alg{\AAIMMFB} and \alg{\AAIMMDAG} algorithms.
We use $\cR=10000$ for the \alg{\AAFF} and \alg{\AAIMMFB} on all datasets due to the high-cost computing resource and memory usage.
Note that, due to the high memory cost, \alg{\AAFF} can not finish the whole test even $\cR = 10000$.
In all tests, we set the seed set $k = 50, 100, 200, 300$ for DM, Flixster, NetHEPT and DBLP respectively, and we also test different combinations of $q_N$ and $q_E$.
For clear representation, we leave out "{\sf AA-}" in the algorithms' names in Figure \ref{fig:budget} and Figure \ref{fig:time}.

\paragraph{Influence Spread Performance.}
From Figure \ref{fig:budget}, it is not hard to see that the influence of the selected seed set is decreasing while increasing either the node budget $q_N$ or the edge budget $q_E$ over all 16 tasks.
On DBLP, we can see that \alg{\AAIMMFB} > \alg{\AAIMMDAG} > \alg{\AAIMMDAG} based on the performance on influence reduction, and \alg{\AAFF} can not finish the test with $\cR = 10000$.
NetHEPT has similar ranking results as DBLP that we have \alg{\AAIMMFB} = \alg{\AAFF} > \alg{\AAIMMNaive} > \alg{\AAIMMDAG} in most cases.
However, on Flixster, the ranking becomes very different. \alg{\AAIMMDAG} becomes the best in most cases, and \alg{\AAFF} and \alg{\AAIMMFB} show the worst influence reduction performance in Figures \ref{fig:budget} (c) and (g).
On DM, all methods perform close to each other.
In summary, the results demonstrate that all algorithms perform close to each other. In general, either \alg{\AAIMMDAG} and \alg{\AAFF} can achieve the best performance in different tasks.
In this case, the running time becomes the key while applying the algorithms in real applications.

\begin{figure}[h]
	\centering
	\subfloat[DBLP ($k=200$, $q_N=30$)]{\includegraphics[width=1.35in]{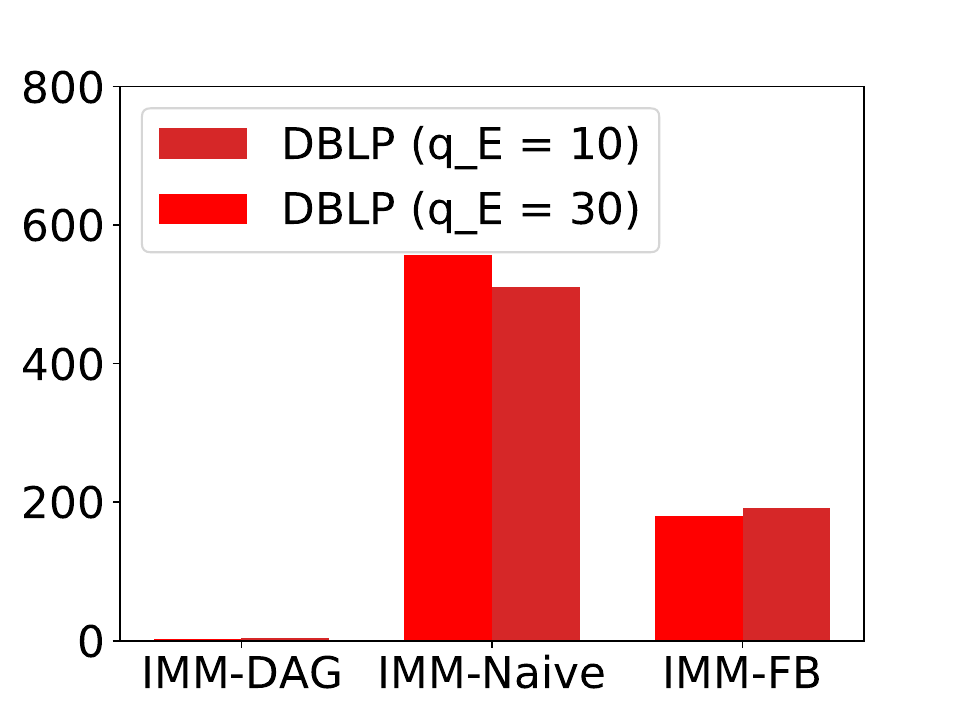}}\ \ \
	\subfloat[NetHEPT ($k=200$, $q_N=30$)]{\includegraphics[width=1.35in]{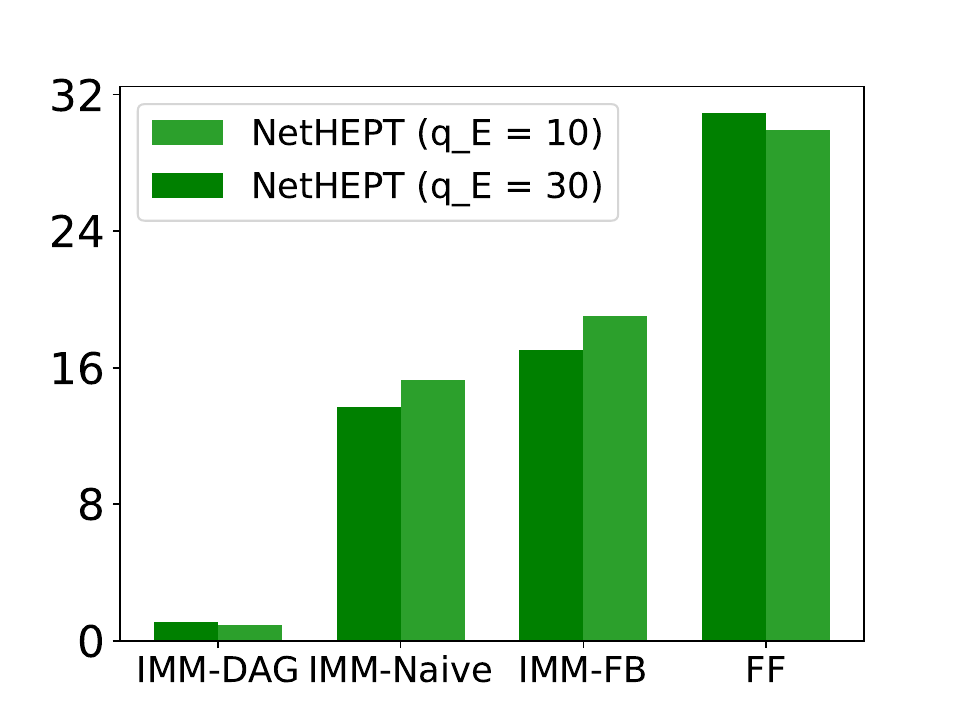}}\ \ \
	\subfloat[Flixter ($k=100$, $q_N=30$)]{\includegraphics[width=1.35in]{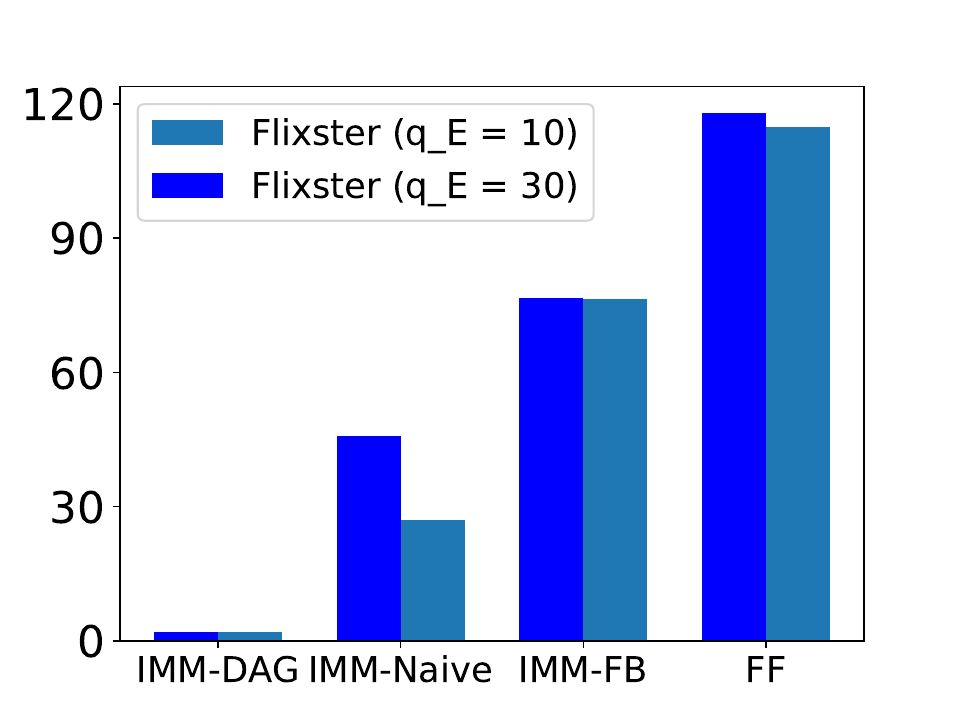}}\ \ \
	\subfloat[DM ($k=50$, $q_N=30$)]{\includegraphics[width=1.35in]{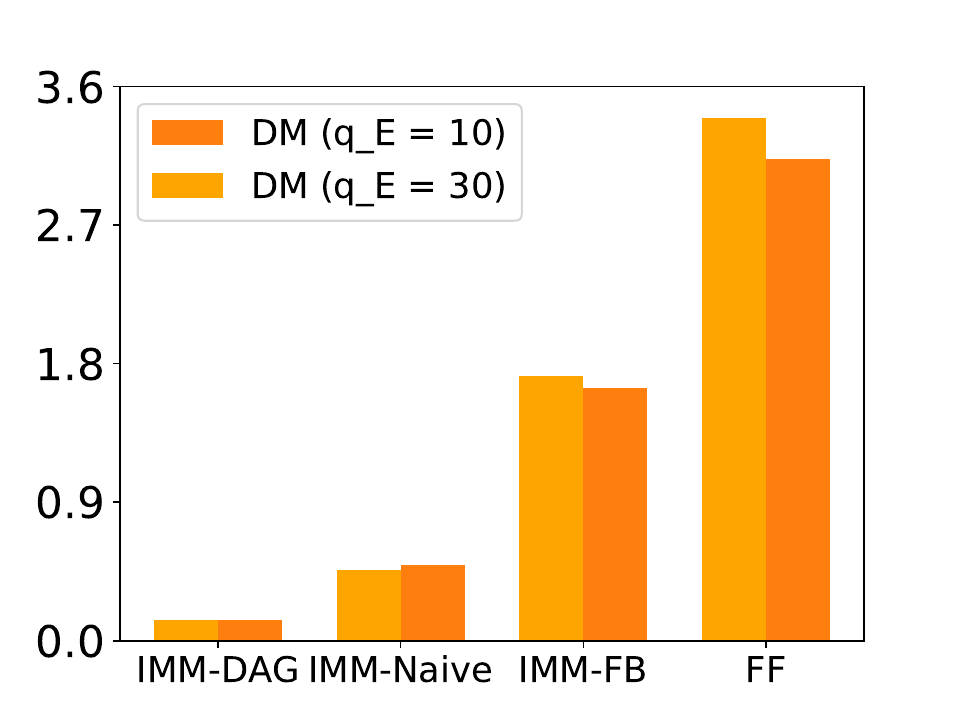}}\\
	\caption{Running time analysis. For all datasets, the node budget is fixed as $30$. The y-axis is the running time (seconds).} \label{fig:time}
	\vspace{-10pt}
\end{figure}

\paragraph{Running time.}
Figure \ref{fig:time} reports the running time of all the tested algorithms on the four datasets.
One clear conclusion is that all IMM algorithms are much more efficient than \alg{\AAFF} that even fails to finish the test on DBLP.
Among three RIS algorithms, it is obvious that \alg{\AAIMMDAG} is the fastest algorithm, \alg{\AAIMMNaive} is the second, and the \alg{\AAIMMFB} is the slowest one.
However, from DBLP's results, \alg{\AAIMMNaive} is slower than \alg{\AAIMMFB}, which is not  due to the limited number of simulations of the \alg{\AAIMMFB}, i.e., $\cR = 10000$.
More importantly, \alg{\AAIMMDAG} is at least 10 times faster than all other algorithms, including other RIS-based approaches.
After combining the influence spread and running time performance together, we recommend \alg{\AAIMMDAG} and \alg{\AAIMMFB} algorithms to be the best two choices for AdvIM, and \alg{\AAFF} to be the worst choice due to the cost of high computer memory and high running time.

\subsection{Discussion}
From the experiments, we found several interesting aspects. 
The most important is why \alg{\AAFF} algorithm took so much memory and computing resource compared to VRR path simulation approaches.
The primary reason is that \alg{\AAFF} algorithm takes too much memory space. For RR path simulation, no matter how big the original seed set is, we only save one single path per simulation.
However, in each \alg{\AAFF} simulation, the size of the forest is related not only to the graph's proprieties but also to the size of the target seed set.
While the target seed set contains 100 nodes, there are 100 sub-tree in each simulation by the \alg{\AAFF} algorithm.
It means \alg{\AAFF} can not be practical in a real application. For example, if we want to stop COVID-19 with a virus spread graph, the seed sets may be thousands in a vast network.
In this case, VRR-based path simulation is the best for a large graph simulation for Adv-IM.
	
%
%

\section{Conclusion} \label{sec:conclude}
In this paper, we study the adversarial attack on influence maximization (AdvIM) task and propose efficient algorithms for solving AdvIM problems.
We adapt the RIS approach to improve the efficiency for the AdvIM task.
The experimental results demonstrate that our algorithms are more effective and efficient than previous approaches.
There are several future directions from this research. One direction is to attack the influence maximization on uncertainty networks or dynamic networks.
We can also study blocking the influence propagation without knowing the seed set. 
Another direction is to sample less number of forward forests with theoretical analysis.
Adversarial attack on other influence propagation models may also be explored.

\newpage
\bibliographystyle{ACM-Reference-Format}
\bibliography{sample-base}


\begin{thebibliography}{28}


\ifx \showCODEN    \undefined \def \showCODEN     #1{\unskip}     \fi
\ifx \showDOI      \undefined \def \showDOI       #1{#1}\fi
\ifx \showISBNx    \undefined \def \showISBNx     #1{\unskip}     \fi
\ifx \showISBNxiii \undefined \def \showISBNxiii  #1{\unskip}     \fi
\ifx \showISSN     \undefined \def \showISSN      #1{\unskip}     \fi
\ifx \showLCCN     \undefined \def \showLCCN      #1{\unskip}     \fi
\ifx \shownote     \undefined \def \shownote      #1{#1}          \fi
\ifx \showarticletitle \undefined \def \showarticletitle #1{#1}   \fi
\ifx \showURL      \undefined \def \showURL       {\relax}        \fi
\providecommand\bibfield[2]{#2}
\providecommand\bibinfo[2]{#2}
\providecommand\natexlab[1]{#1}
\providecommand\showeprint[2][]{arXiv:#2}

\bibitem[Barbieri et~al\mbox{.}(2012)]%
        {barbieri2012topic}
\bibfield{author}{\bibinfo{person}{Nicola Barbieri}, \bibinfo{person}{Francesco
  Bonchi}, {and} \bibinfo{person}{Giuseppe Manco}.}
  \bibinfo{year}{2012}\natexlab{}.
\newblock \showarticletitle{Topic-aware social influence propagation models}.
  In \bibinfo{booktitle}{\emph{Proceeding of the 12th {IEEE} {ICDM}
  International Conference on Data Mining}}. \bibinfo{pages}{81--90}.
\newblock


\bibitem[Bhagat et~al\mbox{.}(2012)]%
        {bhagat_2012_maximizing}
\bibfield{author}{\bibinfo{person}{Smriti Bhagat}, \bibinfo{person}{Amit
  Goyal}, {and} \bibinfo{person}{Laks V.~S. Lakshmanan}.}
  \bibinfo{year}{2012}\natexlab{}.
\newblock \showarticletitle{{Maximizing product adoption in social networks}}.
  In \bibinfo{booktitle}{\emph{Proceedings of the Fifth {WSDM} International
  Conference on Web Search and Web Data Mining}}. \bibinfo{pages}{603--612}.
\newblock


\bibitem[Borgs et~al\mbox{.}(2014)]%
        {BorgsBrautbarChayesLucier}
\bibfield{author}{\bibinfo{person}{Christian Borgs}, \bibinfo{person}{Michael
  Brautbar}, \bibinfo{person}{Jennifer Chayes}, {and} \bibinfo{person}{Brendan
  Lucier}.} \bibinfo{year}{2014}\natexlab{}.
\newblock \showarticletitle{Maximizing social influence in nearly optimal
  time}. In \bibinfo{booktitle}{\emph{Proceedings of the Twenty-Fifth Annual
  {SIAM} {SODA} Symposium on Discrete Algorithms}}. \bibinfo{pages}{946--957}.
\newblock


\bibitem[Budak et~al\mbox{.}(2011)]%
        {BAA11}
\bibfield{author}{\bibinfo{person}{Ceren Budak}, \bibinfo{person}{Divyakant
  Agrawal}, {and} \bibinfo{person}{Amr~El Abbadi}.}
  \bibinfo{year}{2011}\natexlab{}.
\newblock \showarticletitle{{Limiting the spread of misinformation in social
  networks}}. In \bibinfo{booktitle}{\emph{Proceedings of the 20th {WWW}
  International Conference on World Wide Web}}. \bibinfo{pages}{665--674}.
\newblock


\bibitem[Chen(2008)]%
        {chen08}
\bibfield{author}{\bibinfo{person}{Ning Chen}.}
  \bibinfo{year}{2008}\natexlab{}.
\newblock \showarticletitle{On the approximability of influence in social
  networks}. In \bibinfo{booktitle}{\emph{Proceedings of the Nineteenth Annual
  {SIAM} {SODA} Symposium on Discrete Algorithms}}.
  \bibinfo{pages}{1029--1037}.
\newblock


\bibitem[Chen et~al\mbox{.}(2011)]%
        {chen2011influence}
\bibfield{author}{\bibinfo{person}{Wei Chen}, \bibinfo{person}{Alex Collins},
  \bibinfo{person}{Rachel Cummings}, \bibinfo{person}{Te Ke},
  \bibinfo{person}{Zhenming Liu}, \bibinfo{person}{David Rincon},
  \bibinfo{person}{Xiaorui Sun}, \bibinfo{person}{Yajun Wang},
  \bibinfo{person}{Wei Wei}, {and} \bibinfo{person}{Yifei Yuan}.}
  \bibinfo{year}{2011}\natexlab{}.
\newblock \showarticletitle{Influence maximization in social networks when
  negative opinions may emerge and propagate}. In
  \bibinfo{booktitle}{\emph{Proceedings of the Eleventh {SIAM} {SDM}
  International Conference on Data Mining}}. \bibinfo{pages}{379--390}.
\newblock


\bibitem[Chen et~al\mbox{.}(2016)]%
        {ChenLTZZ16}
\bibfield{author}{\bibinfo{person}{Wei Chen}, \bibinfo{person}{Tian Lin},
  \bibinfo{person}{Zihan Tan}, \bibinfo{person}{Mingfei Zhao}, {and}
  \bibinfo{person}{Xuren Zhou}.} \bibinfo{year}{2016}\natexlab{}.
\newblock \showarticletitle{Robust Influence Maximization}. In
  \bibinfo{booktitle}{\emph{Proceedings of the 22nd {ACM} {SIGKDD}
  International Conference on Knowledge Discovery and Data Mining}}.
  \bibinfo{pages}{795--804}.
\newblock


\bibitem[Chen et~al\mbox{.}(2009)]%
        {ChenWY09}
\bibfield{author}{\bibinfo{person}{Wei Chen}, \bibinfo{person}{Yajun Wang},
  {and} \bibinfo{person}{Siyu Yang}.} \bibinfo{year}{2009}\natexlab{}.
\newblock \showarticletitle{Efficient influence maximization in social
  networks}. In \bibinfo{booktitle}{\emph{Proceedings of the 15th {ACM}
  {SIGKDD} International Conference on Knowledge Discovery and Data Mining}}.
  \bibinfo{pages}{199--208}.
\newblock


\bibitem[Chen et~al\mbox{.}(2010)]%
        {ChenYZ10}
\bibfield{author}{\bibinfo{person}{Wei Chen}, \bibinfo{person}{Yifei Yuan},
  {and} \bibinfo{person}{Li Zhang}.} \bibinfo{year}{2010}\natexlab{}.
\newblock \showarticletitle{Scalable Influence Maximization in Social Networks
  under the Linear Threshold Model}. In \bibinfo{booktitle}{\emph{Proceeding of
  the 10th {IEEE} {ICDM} International Conference on Data Mining}}.
  \bibinfo{pages}{88--97}.
\newblock


\bibitem[Domingos and Richardson(2001)]%
        {domingos01}
\bibfield{author}{\bibinfo{person}{Pedro Domingos} {and}
  \bibinfo{person}{Matthew Richardson}.} \bibinfo{year}{2001}\natexlab{}.
\newblock \showarticletitle{Mining the network value of customers}. In
  \bibinfo{booktitle}{\emph{Proceedings of the seventh {ACM} {SIGKDD}
  international conference on Knowledge discovery and data mining}}.
  \bibinfo{pages}{57--66}.
\newblock


\bibitem[Fisher et~al\mbox{.}(1978)]%
        {fisher1978analysis}
\bibfield{author}{\bibinfo{person}{Marshall~L Fisher},
  \bibinfo{person}{George~L Nemhauser}, {and} \bibinfo{person}{Laurence~A
  Wolsey}.} \bibinfo{year}{1978}\natexlab{}.
\newblock \showarticletitle{An analysis of approximations for maximizing
  submodular set functions II}.
\newblock In \bibinfo{booktitle}{\emph{Polyhedral combinatorics}}.
  \bibinfo{publisher}{Springer}.
\newblock


\bibitem[He et~al\mbox{.}(2016)]%
        {he2016joint}
\bibfield{author}{\bibinfo{person}{Lifang He}, \bibinfo{person}{Chun-Ta Lu},
  \bibinfo{person}{Jiaqi Ma}, \bibinfo{person}{Jianping Cao},
  \bibinfo{person}{Linlin Shen}, {and} \bibinfo{person}{Philip~S Yu}.}
  \bibinfo{year}{2016}\natexlab{}.
\newblock \showarticletitle{Joint community and structural hole spanner
  detection via harmonic modularity}. In \bibinfo{booktitle}{\emph{Proceedings
  of the 22nd {ACM} {SIGKDD} International Conference on Knowledge Discovery
  and Data Mining}}. \bibinfo{pages}{875--884}.
\newblock


\bibitem[He and Kempe(2016)]%
        {HeKempe16}
\bibfield{author}{\bibinfo{person}{Xinran He} {and} \bibinfo{person}{David
  Kempe}.} \bibinfo{year}{2016}\natexlab{}.
\newblock \showarticletitle{Robust Influence Maximization}. In
  \bibinfo{booktitle}{\emph{Proceedings of the 22nd {ACM} {SIGKDD}
  International Conference on Knowledge Discovery and Data Mining}}.
  \bibinfo{pages}{885--894}.
\newblock


\bibitem[He et~al\mbox{.}(2012)]%
        {HeSCJ12}
\bibfield{author}{\bibinfo{person}{Xinran He}, \bibinfo{person}{Guojie Song},
  \bibinfo{person}{Wei Chen}, {and} \bibinfo{person}{Qingye Jiang}.}
  \bibinfo{year}{2012}\natexlab{}.
\newblock \showarticletitle{{Influence Blocking Maximization in Social Networks
  under the Competitive Linear Threshold Model}}. In
  \bibinfo{booktitle}{\emph{Proceedings of the Twelfth {SIAM} {SDM}
  International Conference on Data Mining}}. \bibinfo{pages}{463--474}.
\newblock


\bibitem[Jung et~al\mbox{.}(2012)]%
        {JungHC12}
\bibfield{author}{\bibinfo{person}{Kyomin Jung}, \bibinfo{person}{Wooram Heo},
  {and} \bibinfo{person}{Wei Chen}.} \bibinfo{year}{2012}\natexlab{}.
\newblock \showarticletitle{{IRIE: Scalable and Robust Influence Maximization
  in Social Networks}}. In \bibinfo{booktitle}{\emph{Procedding of the 12th
  {IEEE} {ICDM} International Conference on Data Mining}}.
  \bibinfo{pages}{918--923}.
\newblock


\bibitem[Kempe et~al\mbox{.}(2003)]%
        {kempe03}
\bibfield{author}{\bibinfo{person}{David Kempe}, \bibinfo{person}{Jon~M.
  Kleinberg}, {and} \bibinfo{person}{{\'E}va Tardos}.}
  \bibinfo{year}{2003}\natexlab{}.
\newblock \showarticletitle{Maximizing the spread of influence through a social
  network}. In \bibinfo{booktitle}{\emph{Proceedings of the Ninth {ACM}
  {SIGKDD} International Conference on Knowledge Discovery and Data Mining}}.
  \bibinfo{pages}{137--146}.
\newblock


\bibitem[Khalil et~al\mbox{.}(2014)]%
        {khalil2014scalable}
\bibfield{author}{\bibinfo{person}{Elias~Boutros Khalil},
  \bibinfo{person}{Bistra Dilkina}, {and} \bibinfo{person}{Le Song}.}
  \bibinfo{year}{2014}\natexlab{}.
\newblock \showarticletitle{Scalable diffusion-aware optimization of network
  topology}. In \bibinfo{booktitle}{\emph{The 20th {ACM} {SIGKDD} International
  Conference on Knowledge Discovery and Data Mining}}.
  \bibinfo{pages}{1226--1235}.
\newblock


\bibitem[Leskovec et~al\mbox{.}(2007)]%
        {Leskovec07}
\bibfield{author}{\bibinfo{person}{Jure Leskovec}, \bibinfo{person}{Andreas
  Krause}, \bibinfo{person}{Carlos Guestrin}, \bibinfo{person}{Christos
  Faloutsos}, \bibinfo{person}{Jeanne~M. VanBriesen}, {and}
  \bibinfo{person}{Natalie~S. Glance}.} \bibinfo{year}{2007}\natexlab{}.
\newblock \showarticletitle{Cost-effective outbreak detection in networks}. In
  \bibinfo{booktitle}{\emph{Proceedings of the 13th {ACM} {SIGKDD}
  International Conference on Knowledge Discovery and Data Mining}}.
  \bibinfo{pages}{420--429}.
\newblock


\bibitem[Lu et~al\mbox{.}(2015)]%
        {lu2015competition}
\bibfield{author}{\bibinfo{person}{Wei Lu}, \bibinfo{person}{Wei Chen}, {and}
  \bibinfo{person}{Laks~VS Lakshmanan}.} \bibinfo{year}{2015}\natexlab{}.
\newblock \showarticletitle{From competition to complementarity: comparative
  influence diffusion and maximization}.
\newblock \bibinfo{journal}{\emph{Proceedings of the VLDB Endowment}}
  \bibinfo{volume}{9}, \bibinfo{number}{2} (\bibinfo{year}{2015}),
  \bibinfo{pages}{60--71}.
\newblock


\bibitem[Nemhauser et~al\mbox{.}(1978)]%
        {NWF78}
\bibfield{author}{\bibinfo{person}{G.~L. Nemhauser}, \bibinfo{person}{L.~A.
  Wolsey}, {and} \bibinfo{person}{M.~L. Fisher}.}
  \bibinfo{year}{1978}\natexlab{}.
\newblock \bibinfo{booktitle}{\emph{An analysis of the approximations for
  maximizing submodular set functions}}.
\newblock \bibinfo{publisher}{Mathematical Programming}.
\newblock


\bibitem[Nguyen et~al\mbox{.}(2016)]%
        {Nguyen_DSSA_2016}
\bibfield{author}{\bibinfo{person}{Hung~T. Nguyen}, \bibinfo{person}{My~T.
  Thai}, {and} \bibinfo{person}{Thang~N. Dinh}.}
  \bibinfo{year}{2016}\natexlab{}.
\newblock \showarticletitle{Stop-and-Stare: Optimal Sampling Algorithms for
  Viral Marketing in Billion-Scale Networks}. In
  \bibinfo{booktitle}{\emph{Proceedings of the 2016 {ACM} {SIGMOD
  }International Conference on Management of Data}}. \bibinfo{pages}{695--710}.
\newblock


\bibitem[Richardson and Domingos(2002)]%
        {richardson02}
\bibfield{author}{\bibinfo{person}{Matthew Richardson} {and}
  \bibinfo{person}{Pedro Domingos}.} \bibinfo{year}{2002}\natexlab{}.
\newblock \showarticletitle{Mining knowledge-sharing sites for viral
  marketing}. In \bibinfo{booktitle}{\emph{Proceedings of the Eighth {ACM}
  {SIGKDD} International Conference on Knowledge Discovery and Data Mining}}.
  \bibinfo{pages}{61--70}.
\newblock


\bibitem[Tang et~al\mbox{.}(2009)]%
        {tang2009social}
\bibfield{author}{\bibinfo{person}{Jie Tang}, \bibinfo{person}{Jimeng Sun},
  \bibinfo{person}{Chi Wang}, {and} \bibinfo{person}{Zi Yang}.}
  \bibinfo{year}{2009}\natexlab{}.
\newblock \showarticletitle{Social influence analysis in large-scale networks}.
  In \bibinfo{booktitle}{\emph{Proceedings of the 15th {ACM} {SIGKDD}
  International Conference on Knowledge Discovery and Data Mining}}.
  \bibinfo{pages}{807--816}.
\newblock


\bibitem[Tang et~al\mbox{.}(2018)]%
        {Tang_OPIM_2018}
\bibfield{author}{\bibinfo{person}{Jing Tang}, \bibinfo{person}{Xueyan Tang},
  \bibinfo{person}{Xiaokui Xiao}, {and} \bibinfo{person}{Junsong Yuan}.}
  \bibinfo{year}{2018}\natexlab{}.
\newblock \showarticletitle{Online Processing Algorithms for Influence
  Maximization}. In \bibinfo{booktitle}{\emph{Proceedings of the 2018 {ACM}
  {SIGMOD} International Conference on Management of Data}}.
  \bibinfo{pages}{991--1005}.
\newblock


\bibitem[Tang et~al\mbox{.}(2015)]%
        {tang15}
\bibfield{author}{\bibinfo{person}{Youze Tang}, \bibinfo{person}{Yanchen Shi},
  {and} \bibinfo{person}{Xiaokui Xiao}.} \bibinfo{year}{2015}\natexlab{}.
\newblock \showarticletitle{Influence maximization in near-linear time: a
  martingale approach}. In \bibinfo{booktitle}{\emph{Proceedings of the 2015
  {ACM} {SIGMOD} International Conference on Management of Data}}.
  \bibinfo{pages}{1539--1554}.
\newblock


\bibitem[Tang et~al\mbox{.}(2014)]%
        {tang14}
\bibfield{author}{\bibinfo{person}{Youze Tang}, \bibinfo{person}{Xiaokui Xiao},
  {and} \bibinfo{person}{Yanchen Shi}.} \bibinfo{year}{2014}\natexlab{}.
\newblock \showarticletitle{Influence maximization: near-optimal time
  complexity meets practical efficiency}. In
  \bibinfo{booktitle}{\emph{Proceedings of the 2014 {ACM} {SIGMOD}
  International Conference on Management of Data}}. \bibinfo{pages}{75--86}.
\newblock


\bibitem[Tantipathananandh et~al\mbox{.}(2007)]%
        {kempe07}
\bibfield{author}{\bibinfo{person}{Chayant Tantipathananandh},
  \bibinfo{person}{Tanya Berger-Wolf}, {and} \bibinfo{person}{David Kempe}.}
  \bibinfo{year}{2007}\natexlab{}.
\newblock \showarticletitle{A framework for community identification in dynamic
  social networks}. In \bibinfo{booktitle}{\emph{Proceedings of the 13th {ACM}
  {SIGKDD} International Conference on Knowledge Discovery and Data Mining}}.
  \bibinfo{pages}{717--726}.
\newblock


\bibitem[Wang et~al\mbox{.}(2012)]%
        {WCW12}
\bibfield{author}{\bibinfo{person}{Chi Wang}, \bibinfo{person}{Wei Chen}, {and}
  \bibinfo{person}{Yajun Wang}.} \bibinfo{year}{2012}\natexlab{}.
\newblock \showarticletitle{{Scalable influence maximization for independent
  cascade model in large-scale social networks}}.
\newblock \bibinfo{journal}{\emph{Data Mining and Knowledge Discovery (DMKD)}}
  \bibinfo{volume}{25}, \bibinfo{number}{3} (\bibinfo{year}{2012}),
  \bibinfo{pages}{545--576}.
\newblock


\end{thebibliography}

\newpage
\OnlyInFull{
\newpage
\newpage
\appendix

\section*{Appendix}

\section{Forward Forest Algorithm} \label{app:ff}

In this section, we provide a greedy approach with theoretical guarantees by estimating the forward forest simulation, as given in Algorithm \ref{alg:aaff}.
Unlike the naive greedy approaches with Monte Carlo simulations, we further provide the new theoretical analysis based on IMM for \alg{\AAFF}, which can give a more precise and less number of simulations.
Algorithm \ref{alg:aaff} has two phases too. In Phase 1, we generate $\theta$ forward forests $\cF$, where $\theta$ is computed to guarantee the approximation with high probability.
In Phase 2, we use the greedy algorithm to find the $q_N$ nodes and $q_E$ edges that cover the most number of forest forests.
Phase 1 here also follows the IMM structure to estimate a lower bound of the optimal value, which is used to determine the number of the forward forests needed for the theoretical guarantee.
Unlike all RIS approaches, \alg{\AAFF} samples a forest started from the seed set $S$ per time.
In order to save the forward forests, it requires a huge memory compared to the RIS approaches.
So, in practice, we simulate a small number that is far away from the required number to ensure the theoretical guarantee, as shown below.



\begin{restatable}{theorem}{aagreedy} \label{thm:aagreedy}
	Let $A^*$ be the optimal solution of the AdvIM. For every $\varepsilon > 0$ and $\ell > 0$, with probability at least
	$1-\frac{1}{n^\ell}$, the output $A^o$ of the Forward Forest Simulation Adversarial Attacks algorithm \alg{\AAFF}
	satisfies
	\begin{equation*}
	\rho_S(A^o) \geq \left(\frac{1}{2} - \varepsilon\right) \rho_S(A^*),
	\end{equation*}
	In this case, the total running time for \alg{\AAFF} is 
	\begin{align}
	O\left(\frac{(q_N\log{n^-}+q_E\log{m}+\ell\log{n^-})\cdot n}{\OPT\cdot \varepsilon^2} \cdot (\EFF(S) + q \cdot \EFD(S) )\right). \label{eq:FFStime}
	\end{align}
\end{restatable}
\begin{proof}
	First, based on the definition of the $\hat{\rho}(S, \cL)$, we can define the influence deduction function $\rho_S(A)$.
	Because of $|\Gamma(L_i, S) - \Gamma(L_i \setminus A, S)|$ is monotone as a function of $A$, so $\rho_S(A)$ is also monotone.
	In this case, we can apply Lemma \ref{lemma:condition} to ensure the correctness of the \alg{\AAFF} algorithm.
	To ensure the first two conditions of Lemma \ref{lemma:condition},
	we need to generate enough live-edge graphs,
	and the number of the graphs is referred by Lemma \ref{lemma:nonbias}.
	However, Lemma \ref{lemma:nonbias} is two-sided concentration inequality based on Fact \ref{fact:Chernoff}.
	In fact, Lemma \ref{lemma:condition} only need one-sided concentration inequality based on Fact \ref{fact:Chernoff2}.
	So, we can use Lemma \ref{lemma:numtheta} and Fact \ref{fact:Chernoff2}, and we have $\delta_1= \delta_2 = \frac{1}{2n^\ell}$.
	Now, we have a conclusion: let $\alpha = \sqrt{\ell\ln n^- + \ln 2}, \beta = \sqrt{1/2 \cdot(\ln(\binom{n^-}{q_N}\binom{m}{q_E}) + \ell \ln n^-+ ln 2)}$,
	then while $\theta \geq 3n(1/2\cdot \alpha + \beta)^2/(OPT \cdot \varepsilon^2)$,
	two conditions of Lemma \ref{lemma:condition} satisfies, and we have
	\alg{FS-Adv-Greedy} returns $A^o(\cL)$ with at least $1 - \frac{1}{n^\ell}$ probability satisfying $\rho_S(A^o) \geq (\frac{1}{2}-\varepsilon)\rho^*_S(A^*)$.
	
	Next, we prove the time complexity of the \alg{\AAFF}.
	Similar to the above argument, we need 
	\begin{align*}
	O\left( \frac{(q_N\log{n^-}+q_E\log{m}+\ell\log{n^-})\cdot \sigma^-(S)}{\OPT\cdot \varepsilon^2} \right). 
	\end{align*}
	number of forward forests, where $\OPT$ is the optimal solution of the influence reduction problem.
	
	Let $\EFF(S)$ be the mean time of generating a forward forest from seed set $S$, 
	and $\EFD(S)$ be the mean depth of a forward forest from $S$ (depth of a forest is the length of the
	longest path in a forest, and the expectation is taking over the randomness of the forest).
	According the discussion above, average generation of one forest takes $\EFF(S)$ time, and
	update on a forest would take $O(\EFF(S)+q \cdot \EFD(S))$ time, so the total running time is 
	\begin{align*}
	O\left( \frac{(q_N\log{n^-}+q_E\log{m}+\ell)\cdot \sigma^-(S)}{\OPT\cdot \varepsilon^2} \cdot (\EFF(S) + q \cdot \EFD(S) )\right). 
	\end{align*}
	
	Note that $\EFF(S) + q \cdot \EFD(S)$ is dynamic, which depends on the seed set $S$ and the properties of graph $G$.

\end{proof}

\begin{algorithm}[t] 
	\caption{{\AAFF}: Adversarial Attack Forward-Forest Greedy Algorithm
	} \label{alg:aaff}
	\KwIn{Graph $G=(V,E)$, LT influence weights $\{w(u,v)\}_{(u,v)\in E}$, budget $q_N$ for nodes and $q_E$ for edges, 
		seed set $S$
	}
	\KwOut{attack set $A$}
	
	Initialization: attack set $A\leftarrow \emptyset$, forest set $\cF\leftarrow \emptyset$, influence score array $I(a) \leftarrow 0$, for all $a\in V\cup E$\;
	\For{i=1 to $\theta$}{
		Forward propagate a new forest $F_i$ from seed set $S$ 
		by using LT model in $G$ \;
		$\cF \leftarrow \cF \cup F_i$\;
		Calculate the influence score $I_i(a)$ of each node or edge $z$ in $F_i$, Using DFS on $F_i$\;
		Update $I$: for all $a$ in $F_i$, $I(a)\leftarrow I(a) + I_i(a)$;
	}
	\While{$q_N+q_E >0$}{
		\If{$q_E == 0$}{$a \leftarrow  \argmax_{a \in V} I(a)$;}
		\uElseIf{$q_N == 0$}{$a \leftarrow \argmax_{a \in E} I(a)$;}
		\Else{$a \leftarrow  \argmax_{a\in V\cup E} I(a)$;}
		\If{$a \in V$}{$q_N \leftarrow q_N -1$\;}
		\Else{$q_E \leftarrow q_E -1$\;}
		\For{each forest $F_i \in \cF$ containing $a$}{
			Do DFS on $F_i$ starting from $a$, for each node or edge $z$ traversed, 
			calculate its score $I_i(z)$ in $F_i$, update $I(z) \leftarrow I(z) - I_i(z)$\;
			Remove subtree rooted at $a$ from $F_i$\;
			For all ancestor node or edge $z$ of $a$ in $F_i$, update $I(z) \leftarrow I(z) - I_i(a)$\;
		}
		$A \leftarrow A \cup a$\;
	}		
	\bf{return} $A$.
\end{algorithm}
\label{app:proof}

\section{Proof - Supplementary} \label{app:other}



\begin{restatable}{fact}{Chernoff} \label{fact:Chernoff}
	(Chernoff bound) Let $X_1, X_2, \ldots, X_R$ be $R$ independent random variables and with $X_i$ having range [0, 1]. Let $X = \sum_{i=1}^R X_i / R$, for any $0 < \gamma < 1$, we have
	\begin{equation*}
	\Pr\{|X - \E[X| | \geq \gamma \cdot \E[X]\} \leq 2 \exp(-\frac{\gamma^2\E[X]}{3}).
	\end{equation*}
\end{restatable}

\begin{restatable}{fact}{Chernoff2} \label{fact:Chernoff2}
	(Chernoff bound) Let $X_1, X_2, \ldots, X_R$ be $R$ independent random variables and with $X_i$ having range [0, 1]. and there exists $\mu \in [0, 1]$ making $\E[X_i] = \mu$ for any $i \in [R]$. Let $Y = \sum_{i=1}^R X_i$, for any $\gamma> 0$, 
	\begin{equation*}
	\Pr\{Y - t\mu \geq \gamma \cdot t\mu\} \leq \exp(-\frac{\gamma^2}{2 + \frac{2}{3}\gamma}t\mu).
	\end{equation*}
	For any $0 < \gamma < 1$, 
	\begin{equation*}
	\Pr\{Y - t\mu \leq - \gamma \cdot t\mu\} \leq \exp(-\frac{\gamma^2}{2}t\mu).
	\end{equation*}
\end{restatable}

\begin{restatable}{lemma}{condition} \label{lemma:condition}
	For any $\varepsilon > 0, \varepsilon_1 \in (0, \varepsilon/(1 - 1/2))$,
	for any $\delta_1, \delta_2 > 0$, 
	if $(a) \Pr_{\omega \in \Omega}\{\rho_S(A^*, \omega) \geq (1 - \varepsilon_1) \cdot OPT \} \geq 1 - \delta_1$;
	$(b)$ for every bad $A$ to $\varepsilon$, $\Pr_{\omega \in \Omega}\{\rho_S(A \omega) \geq  (1 - 1/2) (1 - \varepsilon_1) \cdot OPT \} \leq \delta_2 / (\binom{n^-}{q_N}\binom{m}{q_E})$;
	$(c)$ for all $\omega \in \Omega$, $\rho_S(A, \omega)$ is submodular to $S$,
	then, with at least $1 - \delta_1 - \delta_2$, the greedy output $A^o(\omega)$ of $\rho_S(\cdot, \omega)$ holds 
	\begin{equation*}
	\Pr_{\omega \in \Omega}\{\rho_S(A^o, \omega) \geq  (1/2 - \varepsilon) \cdot OPT \} \leq 1- \delta_1 - \delta_2, 
	\end{equation*}
\end{restatable}
\begin{proof}
	Because $\rho_S(A^o, \omega)$ is both monotone and submodular, then we have:
	\begin{equation*}
	\rho_S(A^o, \omega) \geq \frac{1}{2} \rho_S(A^*, \omega).
	\end{equation*}
	Based on the condition (a) in this lemma, we know it holds at least $1 - \delta_1$ probability,
	\begin{equation*}
	\rho_S(A^o, \omega) \geq \frac{1}{2} \rho_S(A^*, \omega) \geq \frac{1}{2}(1-\varepsilon_1)\cdot OPT.
	\end{equation*}
	Based on the condition (b) in this lemma, existing a bad $A$ causing $\rho_S(A, \omega) \geq \frac{1}{2}(1-\varepsilon_1)\cdot OPT$ with at most $\delta_2$ probability.
	Because of at most $\binom{n^-}{q_N}\binom{m}{q_E}$ sets with $q_N$ nodes and $q_E$ edges, for every bad $S$,
	the probability would not be more than $\delta_2 / (\binom{n^-}{q_N}\binom{m}{q_E})$.
	Then, based on the union bound, we know that, with at least $1 - \delta_1 - \delta_2$ probability, $\rho_S(A^o)$ is not bad, which means $	\rho_S(A^o, \omega) \geq (1/2 -\varepsilon) \cdot OPT$.
\end{proof}

	

\begin{restatable}{lemma}{nonbias} \label{lemma:nonbias}
	For any attack set $A$, the original greedy algorithm returns $\hat{\rho}_S(A)$ non-bias influence estimation $\rho_S(A)$.
	For any $0 < \delta < 1$ and $0 < \gamma < 1$, while $R \geq \frac{3n^-}{\gamma^2\rho_S(A)} \ln(\frac{2}{\delta})$, we have
	$\Pr\{|\hat{\rho}_S(A) - \rho_S(A)|\leq \gamma\cdot \rho_S(A)\} \geq 1 - \delta$.
	If $f_v(A)$ spends constant time, the time complexity is $O(n^-R)$.
\end{restatable}
\begin{proof}
    Let $X_i$ become the percentage of total removed nodes among all nodes in the $i$-th simulation. Then, we have (a) $X_i \in [0, 1]$; (b) $\E[X_i] = \rho_S(A) / n^-$; and (c) $\hat{\rho}_S(A) = \sum_{i=1}^R n^- X_i / R$. Let $X = \sum_{i=1}^R X_i= \hat{\rho}_S(A) \cdot R / n^-$. First,
	\begin{equation*}
	\E[\hat{\rho}_S(A)] = \E \left[\frac{\sum_{i=1}^R n^- X_i}{R}\right] = \frac{\sum_{i=1}^R n^- \E [X_i]}{R} = \rho_S(A).
	\end{equation*}
	Then, $\hat{\rho}_S(A)$ is non-bias estimation $\rho_S(A)$.
	Next, due to simulations being independent between each other, based on the Chernoff bound (Fact \ref{fact:Chernoff}), we have
	\begin{align*}
	& \Pr\{|\hat{\rho}_S(A) - \rho_S(A)| \leq \gamma \cdot \rho_S(A)\} \\
	&= \Pr\{|R \cdot \hat{\rho}_S(A) / n^- - R \cdot \rho_S(A) / n^-| \leq \gamma \cdot R \cdot \rho_S(A) / n^- \} \\
	&= \Pr\{|X - \E[X]| \leq \gamma \cdot \E[X]\} \\
	&= 1 - \Pr\{|X - \E[X]| > \gamma \cdot \E[X]\} \\
	&\geq 1 - 2 \exp(-\frac{\gamma^2\E[X]}{3}) = 1 - 2 \exp(-\frac{\gamma^2 R\cdot \rho_S(A)}{3})\\
	\end{align*}
	The last step is based on Chernoff bound. Finally, let $R \geq -\frac{3 n^-}{\gamma^2 \rho_S(A)} \ln(\frac{2}{\delta})$, we have $\Pr\{|\hat{\rho}_S(A) - \rho_S(A)|\leq \gamma\cdot \rho_S(A)\} \geq 1 - \delta$.
\end{proof}

\begin{restatable}{lemma}{numtheta} \label{lemma:numtheta}
	\cite{tang15}. For any $\varepsilon > 0$, $\varepsilon_1 \in (0, \varepsilon/(1-1/2))$, for any $\delta_1,\delta_2 > 0$, let
	\begin{align*}
	& \theta_1 = \frac{2n^- \cdot \ln(1/\delta_1)}{OPT \cdot \varepsilon_1^2},  \theta_2 = \frac{ n^- \cdot \ln(\binom{n}{q_N}\binom{m}{q_E}/\delta_2)}{OPT \cdot (\varepsilon - 1/2 \cdot \varepsilon_1)^2}.
	\end{align*}
	For any fixed $\theta \geq \theta_1$, we have 
	\begin{align*}
	    \Pr_{\omega \in \Omega}\{\rho_S(A^*, \omega) \geq  (1 - \varepsilon_1) \cdot OPT \} \geq 1- \delta_1.
	\end{align*}
	For any fixed $\theta \geq \theta_2$ and for every bad $A$ to $\varepsilon$, we have 
	\begin{align*}
	    \Pr_{\omega \in \Omega}\{\rho_S(A, \omega) \geq  (1 - 1/2)(1 - \varepsilon_1) \cdot OPT \} \leq \frac{\delta_2}{\binom{n}{q_N}\binom{m}{q_E}}.
	\end{align*}
\end{restatable}

\begin{proof}
     Let $X_i$ become the percentage of total removed nodes among all nodes in the $i$-th simulation. Based on the Chernoff bound (Fact \ref{fact:Chernoff2}), we have
     \begin{align*}
         & \Pr_{\omega \in \Omega} \{ \rho_S(A^*, \omega) < (1 - \varepsilon_1) \cdot OPT \} \\
         & = \Pr_{\omega \in \Omega} \{ n^- \cdot \sum_{i=1}^\theta X_i(A^*) / \theta < (1-\varepsilon_1) \cdot \rho_S(A^*) \} \\
         & = \Pr_{\omega \in \Omega} \{ \sum_{i=1}^\theta X_i(A^*) - \theta \cdot \rho_S(A^*)/n^- < -\varepsilon_1 \cdot (\theta \cdot \rho_S(A^*))/n^- \} \\
         & \leq exp(-\frac{\varepsilon_1^2}{2n^-}\cdot \theta \rho_S(A^*) ) \\ 
         & \leq exp(-\frac{\varepsilon_1^2}{2n^-}\cdot \frac{2n^- \cdot \ln(1/\delta_1)}{OPT \cdot \varepsilon_1^2} \cdot \rho_S(A^*) ) \\
         & = \delta_1.
     \end{align*}
     Therefore, 
     \begin{align*}
	    \Pr_{\omega \in \Omega}\{\rho_S(A^*, \omega) \geq  (1 - \varepsilon_1) \cdot OPT \} \geq 1- \delta_1.
	\end{align*}

    Let $\varepsilon_2 = \varepsilon - (1-1/2)\cdot \varepsilon_1$, then for every bad $A$ to $\varepsilon$ ({\em i.e.} $\rho_S(A) < (1/2-\varepsilon)\cdot OPT$, we have
    \begin{align*}
        & \Pr_{\omega \in \Omega}\{\rho_S(A, \omega) \geq  (1 - 1/2)(1 - \varepsilon_1) \cdot OPT \} \\
        & = \Pr_{\omega \in \Omega}\{ \sum_{i=1}^\theta X_i(A) - \theta\cdot \frac{\rho_S(A)}{n^-} \geq \frac{\theta}{n^-} \cdot \big( \frac{1}{2} (1-\varepsilon_1)\cdot OPT - \rho_S(A) \big)\} \\ 
        & \leq \Pr_{\omega \in \Omega}\{ \sum_{i=1}^\theta X_i(A) - \theta\cdot \frac{\rho_S(A)}{n^-} \geq \frac{\theta}{n^-}\cdot \varepsilon_2 \cdot OPT \} \\ 
        & = \Pr_{\omega \in \Omega} \{ \sum_{i=1}^\theta X_i(A) - \theta\cdot \frac{\rho_S(A)}{n^-} \geq (\varepsilon_2 \cdot \frac{OPT}{\rho_S(A)}) \cdot \theta \cdot \frac{\rho_S(A)}{n^-} \} \\
        & \leq exp\left( -\frac{(\varepsilon_2 \cdot \frac{OPT}{\rho_S(A)})^2}{2 + \frac{2}{3}(\varepsilon_2 \cdot \frac{OPT}{\rho_S(A)})} \cdot \theta \cdot \frac{\rho_S(A)}{n^-} \right) \\ 
        & \leq exp\left( -\frac{\varepsilon_2^2 \cdot OPT^2}{2\rho_S(A) + \frac{2}{3}(\varepsilon_2 \cdot OPT)} \cdot \frac{\theta}{n^-} \right) \\ 
        & \leq exp\left( -\frac{\varepsilon_2^2 \cdot OPT^2}{2(1-1/2-\varepsilon)\cdot OPT + \frac{2}{3}(\varepsilon_2 \cdot OPT)} \cdot \frac{\theta}{n^-} \right) \\ 
        & \leq exp\left( -\frac{(\varepsilon-(1-1/2)\cdot\varepsilon_1)^2 \cdot OPT}{2(1-1/2)} \cdot \frac{\theta}{n^-} \right) \\ 
        & \leq exp\left( -(\varepsilon-1/2\cdot\varepsilon_1)^2 \cdot OPT \cdot \frac{ n^- \cdot \ln(\binom{n}{q_N}\binom{m}{q_E}/\delta_2)}{OPT \cdot (\varepsilon - 1/2 \cdot \varepsilon_1)^2} \cdot \frac{1}{n^-} \right) \\
        & = \delta_2/(\binom{n}{q_N}\binom{m}{q_E}).
    \end{align*}
\end{proof}

}


\end{document}